\documentclass{article}
\pdfoutput=1

\usepackage[textwidth=151.7mm]{geometry}
\usepackage{amssymb, mathtools, bm, dsfont}
\usepackage{pgfplots}
\usepackage{array}
\usepackage{amsthm, thm-restate}
\usepackage{subcaption}
\usepackage{xstring}
\usepackage[pdfauthor={Sven Jäger}, pdftitle={An Improved Greedy Algorithm for Stochastic Online Scheduling on Unrelated Machines}, pdfkeywords={Online Scheduling, Stochastic Scheduling, Total Weighted Completion Time, Unrelated Machines}]{hyperref}
\usepackage[noabbrev, capitalise]{cleveref}
\usepackage[longnamesfirst,numbers]{natbib}

\newtheorem{theorem}{Theorem}
\newtheorem{lemma}[theorem]{Lemma}
\newtheorem{proposition}[theorem]{Proposition}
\newtheorem{corollary}[theorem]{Corollary}

\newcommand{\Z}{\mathbb Z}
\newcommand{\N}{\mathbb N}
\newcommand{\R}{\mathbb R}
\newcommand{\E}{\mathrm E}
\newcommand{\Var}{\mathrm{Var}}
\newcommand{\CV}{\mathrm{CV}}
\DeclarePairedDelimiterX\condexpdelim[2]{[}{]}{#1\mathclose{}\, \delimsize\vert\,\mathopen{}#2}
\newcommand{\condexp}[3][]{\E\condexpdelim[#1]{#2}{#3}}
\newcommand{\dif}{\mathop{}\mathrm{d}}
\DeclareMathOperator{\cost}{cost}
\newcommand{\ineq}[1]{inequality~\labelcref{#1}}

\definecolor{myred}{HTML}{DB001F}
\definecolor{myorange}{HTML}{F68C21}
\definecolor{mygreen}{HTML}{238224}
\definecolor{myblue}{HTML}{3840FE}
\definecolor{mypurple}{HTML}{730087}
\pgfplotsset{
 compat=1.17, height=7cm, width=.65\textwidth, axis x line=bottom, axis y line=left, clip mode=individual, samples=1000, clip=false,
 axis line style={->, thick},
 x label style={at={(ticklabel* cs:1)}, anchor=west},
 y label style={at={(ticklabel* cs:0.99)}, rotate=-90, anchor=west}
}

\newcommand{\thirdoffive}[5]{#3}
\newcommand{\lcthirdoffive}[5]{\MakeLowercase{#3}}
\newcommand{\lcnameref}[1]{\begingroup\let\thirdoffive\lcthirdoffive\nameref{#1}\endgroup}
\crefformat{appendix}{Appendix~#2\StrRight{#1}{1}#3}

\title{An Improved Greedy Algorithm for Stochastic Online Scheduling on Unrelated Machines}

\author{Sven Jäger\thanks{This work was primarily done while affiliated with: Technische Universität Berlin, Straße des 17.~Juni 136, 10623 Berlin, Germany. Present address: Technische Universität Kaiserslautern, Paul-Ehrlich-Straße 14, 67663 Kaiserslautern, Germany}}

\begin{document}

\maketitle

\begin{abstract}
 Most practical scheduling applications involve some uncertainty about the arriving times and lengths of the jobs. Stochastic online scheduling is a well-established model capturing this. Here the arrivals occur online, while the processing times are random. For this model, \citeauthor{GMUX21} recently devised an efficient policy for non-preemptive scheduling on unrelated machines with the objective to minimize the expected total weighted completion time. We improve upon this policy by adroitly combining greedy job assignment with $\alpha_j$-point scheduling on each machine. In this way we obtain a $(3+\sqrt 5)(2+\Delta)$-competitive deterministic and an $(8+4\Delta)$-competitive randomized stochastic online scheduling policy, where $\Delta$ is an upper bound on the squared coefficients of variation of the processing times. We also give constant performance guarantees for these policies within the class of all fixed-assignment policies. The $\alpha_j$-point scheduling on a single machine can be enhanced when the upper bound $\Delta$ is known a priori or the processing times are known to be $\delta$-NBUE for some $\delta \ge 1$. This implies improved competitive ratios for unrelated machines but may also be of independent interest.
\end{abstract}

\section{Introduction} \label{sec:introduction}

We consider stochastic scheduling on $m$ unrelated machines with the objective of minimizing the expected total weighted completion time of all jobs. Each job has to be scheduled uninterruptedly on one of the machines, and at any point in time, each machine can process at most one job. We assume that the time that machine~$i$ needs to process job~$j$ is a random variable~$\bm p_{ij}$ and that the processing times of different jobs are stochastically independent. The distribution of each $\bm p_{ij}$ is known beforehand, but the actual processing time of job~$j$ becomes only known when it is completed. In addition, each job~$j$ has a deterministic weight~$w_j \ge 0$ and release date~$r_j \ge 0$, before which it cannot be started. The goal is to find a non-anticipative \emph{scheduling policy} that minimizes the expected total weighted completion time for the given processing time distributions. Any decision of such a policy at some time~$t$ must depend only on the input data and the information about processing times gathered up to time~$t$. A formal definition of a scheduling policy is given by \citet{MRW84}; see also \cite{MR85}. According to this definition, a scheduling policy is, in analogy to a schedule, specific to an instance of the stochastic scheduling problem. In contrast, we regard a scheduling policy as a general rule that works for an arbitrary number of jobs and meets the technical requirements of \citeauthor{MRW84}'s definition for every possible problem instance. In the $3$-field notation the considered problem is denoted as $\mathrm R \mid \bm p_{ij} \sim \mathrm{stoch},\ r_j \mid \E[\sum w_j \bm C_j]$. A special class of scheduling policies are \emph{fixed-assignment policies} whose assignment of jobs to machines is independent of the processing time realizations.

Stochastic processing times model one aspect of uncertainty found in many practical applications. Another aspect in most real-world problems is that the scheduler does not know jobs to be released in the future. This is captured by the online time model. A scheduling policy whose decisions at any time do not depend on jobs released after that time is called a stochastic online scheduling policy or, for short, \emph{online policy}. We additionally demand that jobs cannot be hold in a central queue but be immediately assigned to some machine at their arrival. This requirement, known as \emph{immediate dispatch}, occurs in practical applications such as large computing systems~\cite{AA07}.

The performance of an online policy is measured relative to an offline scheduling policy, which knows all job weights, release dates, and processing time distributions from the beginning, but also only learns the actual processing times over time. An online policy~$\Pi$ is called \emph{$\rho$-competitive} if for any problem instance~$(F_{ij}, w_j, r_j)$ and any scheduling policy~$\Pi'$
\[\E_{\bm p_{ij} \sim F_{ij}}\biggl[\sum_j w_j \bm C_j^{\Pi}\biggr] \le \rho \cdot \E_{\bm p_{ij} \sim F_{ij}}\biggl[\sum_j w_j \bm C_j^{\Pi'}\biggr].\] Here $F_{ij}$ denotes the probability distribution of the processing time~$\bm p_{ij}$. We also consider \emph{randomized online policies}, which have access to an arbitrary number of independent random variables~$\bm X = (\bm X_1,\dotsc,\bm X_k)$. Such a policy $\Pi_{\bm X}$ is called \emph{$\rho$-competitive} if for all problem instances~$(F_{ij}, w_j, r_j)$ and any scheduling policy~$\Pi'$
\[\E_{\bm p_{ij} \sim F_{ij},\bm X}\biggl[\sum_j w_j \bm C_j^{\Pi_{\bm X}}\biggr] \le \rho \cdot \E_{\bm p_{ij} \sim F_{ij}}\biggl[\sum_j w_j \bm C_j^{\Pi'}\biggr].\] This definition amounts to the oblivious adversary model for the stochastic scheduling problem.

\subsection{Related work}
\paragraph{Stochastic scheduling}

The first work on analyzing non-optimal scheduling policies with respect to their multiplicative worst-case performance guarantee was done by \citet{MSU99}. They developed a job-based list scheduling policy with performance guarantee~$3-1/m+ \max\{1,\ (1-1/m) \cdot \Delta\}$ for $m$ identical machines, where $\Delta$ is an upper bound on the squared coefficients of variation of the jobs' processing times, i.e., $\Delta \ge \CV[\bm p_j]^2 \coloneqq \Var[\bm p_j]/\E[\bm p_j]^2$ for all processing times~$\bm p_j$. This policy solves a linear programming relaxation, from whose solution the job order used for list scheduling is derived. It needs access to the release dates, weights, and expected processing times of all jobs at the beginning and can therefore not be executed online. The same applies to an efficient randomized fixed-assignment policy for unrelated machines by \citet{SSU16}. This policy randomly rounds an LP solution in order to determine the assignment of jobs to machines---an idea adopted from \citet{SS02}. It has a performance guarantee of $2+\Delta+\varepsilon$ for arbitrary $\varepsilon > 0$. \Citeauthor{SSU16} also proved that the \emph{adaptivity gap} for this problem, i.e.\ the worst-case ratio of the expected objective values of an optimal fixed-assignment policy and an optimal adaptive policy, is at least $\Delta/2$.

\paragraph{Stochastic online scheduling}
\Citet{MUV04, MUV06} introduced the stochastic online scheduling model considered in this paper. They proved that on a single machine, applying the machine-based \emph{weighted shortest expected processing time} (\emph{WSEPT}) rule to suitably modified release dates is $(2+\delta)$-competitive for \emph{$\delta$-NBUE processing times}, i.e.\ processing times~$\bm p_j$ satisfying $\E[\bm p_j - t \mid \bm p_j > t] \le \delta \cdot \E[\bm p_j]$ for all $t \ge 0$. The idea of list scheduling with modified release dates is borrowed from an online algorithm by \citet{MS04}. When applied to the original release dates, the WSEPT rule is asymptotically optimal when the weights and processing times are uniformly bounded, as shown by \citet{CLQS06}.
For identical machines \citeauthor{MUV06} designed an online fixed-assignment policy that assigns each job greedily to a machine minimizing the increase of some lower bound on the expected total weighted completion time. The jobs assigned to each machine are then scheduled according to the machine-based WSEPT rule with modified release dates. The competitive ratio of this policy is bounded by $1+\max\{1 + \delta/\alpha,\ \alpha + (m-1)/(2m) \cdot (1+\Delta)\} \le 1 + \max\{1 + \delta/\alpha,\ \alpha + (2m-1)/m \cdot \delta\}$, where $\alpha \ge 0$ can be chosen arbitrarily. Optimizing the second term for $\alpha$ yields the performance guarantee~$3/2 + (2m-1)/(2m) \cdot \delta + \sqrt{4\delta^2+1}/2$.

\Citet{Sch08} transferred the idea of (delayed) list scheduling in order of $\alpha_j$-points, having proved successful in deterministic scheduling~\cite{PSW98, SS02, Sku06}, to stochastic scheduling. Here, the $\alpha_j$-point of a job~$j$ for $\alpha_j \in (0, 1]$ is the first point in time when an $\alpha_j$-fraction of its deterministic counterpart has been processed in a virtual preemptive schedule. \Citeauthor{Sch08} applied this technique either with $\alpha_j$ chosen independently and uniformly at random from $(0, 1]$ or with $\alpha_j \coloneqq \phi - 1$, where $\phi = (1+\sqrt 5)/2$ is the golden ratio. In this way he succeeded in developing the currently best known randomized and deterministic efficient scheduling policies for identical machines. These are $(2+\Delta)$-competitive and $(1 + \max\{\phi, (\phi+1)/2 \cdot (1+\Delta)\})$-competitive, respectively.

An online policy for unrelated machines was found only recently by \citet*{GMUX20, GMUX21}. It is a $(184/51 \cdot (2-g(\Delta)) \cdot (2+\Delta))$-competitive fixed-assignment policy, where
\begin{equation}
 g(\Delta) \coloneqq \begin{cases*} \frac{2-\sqrt{\Delta}}{2} &if $\Delta \le 1$;\\ \frac{1}{\Delta+1} &if $\Delta \ge 1$. \end{cases*} \label{eq:def g}
\end{equation}
Each job is assigned at its release date to a machine minimizing some upper bound on the cost increase incurred by adding this job to a virtual schedule that also takes into account some jobs that might be assigned to it later.\footnote{In an earlier, faulty version a simpler greedy assignment rule was used.} On each machine the jobs are scheduled by delayed job-based list scheduling according to the start times of the machine-based WSPT schedule of their deterministic counterparts.

\paragraph{Other uncertainty models}

A variety of different models for scheduling under uncertainty have evolved, ranging from robust scheduling to random-order models. In this overview we restrict to models with potential uncertainty in the set of jobs arriving and the processing times. Both arrivals and processing times can be either known, stochastic, or adversarial (possibly restricted to a given uncertainty set). Another important feature of a model for uncertainty is the performance criterion used. This can either be absolute or relative to the ex post optimal solution. For example, minimizing the (absolute) worst-case objective value of all possible scenarios amounts to robust scheduling, see e.g.~\cite{BJPR21}, while minimizing the (absolute) expected objective boils down to stochastic scheduling. On the other hand, online algorithms are evaluated based on their (relative) worst-case competitive ratio; see e.g.~\cite{PST04}. For stochastic processing times, a relative performance measure has been introduced by \citet{SSS06}, see also \cite{SS06}. When randomization is allowed, one can distinguish further according to when the random values become known to the adversary (oblivious vs.\ adaptive adversary).

This results in a variety of possible combinations of models of uncertainty used for the processing times and job arrivals. Clearly, any worst-case guarantee for an adversarial model implies the same guarantee for the corresponding stochastic model. In this paper we consider stochastic processing times and adversarial arrivals. Going into the direction of more uncertainty, one arrives at non-clairvoyant online scheduling, which is usually investigated in the preemptive setting; see \cite{IKM18}. The same is also true when more stochastic information is added by replacing the worst-case arrivals by random arrivals; see e.g.~\cite{CCPW92} for identical machines. To our knowledge, no models with random durations and arrivals have been studied for unrelated machines, so that our results are also the best known results in this setting. Finally, the considered model generalizes deterministic online scheduling with immediate dispatch. For this model \citeauthor{GMUX21}'s policy applied to deterministic jobs yields a $7.216$-competitive deterministic online algorithm. When the immediate dispatch requirement is dropped, there is a recent $3$-competitive deterministic online algorithm due to \citet{BKL21}. It requires the repeated optimal solution of a generalization of the offline scheduling problem and is thus not efficient. When randomization is allowed, there are a $(4/\ln 2)$-competitive efficient~\cite{HSSW97, CPS+96} and a $(1+1/\ln(2))$-competitive inefficient online algorithm~\cite{BKL21} without immediate dispatch. On the other hand, there exist lower bounds of $2$ for deterministic online algorithms~\cite{HV96} and of $\mathrm{e}/(\mathrm e - 1)$ for randomized online algorithms~\cite{SV02}, which carry over to the stochastic online model. Theoretically, the remaining gaps can be narrowed down arbitrarily for any fixed number of machines by means of a competitive-ratio approximation scheme due to \citet{LMMW16}.

\subsection{Results and methodology}

Our main contribution is a simple deterministic greedy rule for assigning jobs to unrelated machines that, combined with $\alpha_j$-scheduling on each machine, leads to efficient fixed-assignment online policies with immediate dispatch for non-preemptive stochastic scheduling with the currently best known competitiveness bounds. In this framework, different choices of the $\alpha_j$ are possible. By choosing $\alpha_j$ independently and uniformly at random from $(0, 1]$, i.e., by executing \citeauthor{Sch08}'s RSOS policy on each machine, we obtain a randomized online policy that is $(8+4\Delta)$-competitive within the class of all scheduling policies and $8$-competitive within the class of fixed-assignment policies. Similarly, following \citeauthor{Sch08}'s DSOS rule on each machine results in a deterministic online policy that is $((3+\sqrt 5) (2+\Delta))$-competitive within all scheduling policies and $(2 (3 + \sqrt 5))$-competitive within the class of fixed-assignment policies. 

Since the considered problem has an unbounded adaptivity gap, our fixed-assignment policies cannot have a bounded competitive ratio for arbitrary processing time distributions. However, this is the case when restricting to instances with bounded coefficients of variation. Still, the execution of these policies is possible without a priori knowledge of such a bound and always results in a schedule whose expected cost is bounded in terms of the best possible upper bound---the maximum coefficient of variation of the received instance. However, if a bound on the coefficients of variation is known in advance, the choice of the $\alpha_j$ can be adapted to this, resulting in enhanced competitiveness results. All our bounds as well as the bound due to \citeauthor{GMUX21} are illustrated in \cref{fig:performance unrelated}.
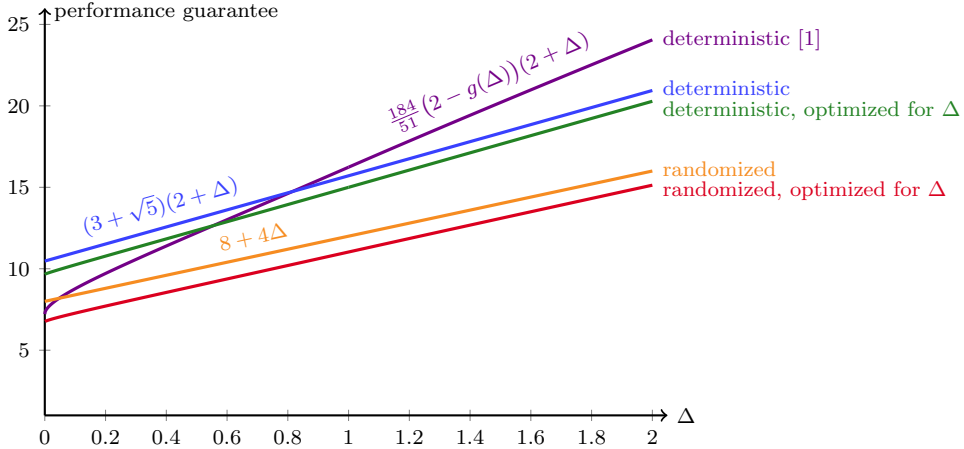
\begin{figure}
 \Crefname{corollary}{Cor.}{Cor.}
 \centering
 \begin{tikzpicture}[font=\footnotesize]
  \begin{axis}[xmin=0, xmax=2.05, ymin=1, ymax=26, xlabel={$\Delta$}, ylabel={performance guarantee}]
   \addplot[very thick, mypurple, domain=0:1] {184/51*(1+sqrt(x)/2)*(2+x)};
   \addplot[very thick, mypurple, domain=1:2] {184/51*(2-1/(x+1))*(2+x)} node[right] {deterministic \cite{GMUX21}}  node[above, sloped, pos=0.5] {$\frac{184}{51}\bigl(2-g(\Delta)\bigr)(2+\Delta)$};
   \addplot[myblue, very thick] coordinates {(0, {2*(3+sqrt(5))}) (2, {4*(3+sqrt(5))})} node[yshift=1pt, right] {deterministic} node [sloped, pos=0.2, above] {$(3+\sqrt 5) (2+\Delta)$};
   \addplot[mygreen, very thick, domain=0:1] (x, {1/2*(2+x)*(4+sqrt(x)+sqrt(32-8*sqrt(x)+x))});
   \addplot[mygreen, very thick, domain=1:2] plot (x, {2*(2+x)*(1+(2*(2+x))/(sqrt(8+12*x+5*x*x)-x))}) node[right, yshift=-3.5pt] {deterministic, optimized for $\Delta$};
   \addplot[myorange, very thick] coordinates {(0, 8) (2, 16)} node [yshift=1pt, right] {randomized} node [above, sloped, pos=0.35] {$8+4\Delta$};
   \addplot[very thick, myred] table[x expr=(\coordindex+1)/1000, y index=0] {performance_rand.dat} node[right, yshift=-2pt] {randomized, optimized for $\Delta$};
  \end{axis}
 \end{tikzpicture}
 \caption{Performance guarantees of different online policies with immediate dispatch for unrelated machines within the class of all scheduling policies as a function of upper bound~$\Delta$ for the squared coefficients of variation.}
\label{fig:performance unrelated}
\end{figure}

The approach to obtain these results can be described on a high level as follows: Whenever a job is released, it is assigned to a machine minimizing the immediate increase of a specific easily computable surrogate cost function that approximates the expected sum of weighted completion times of the jobs assigned to the machine. The sum of the surrogate costs resulting from this greedy assignment is compared to the optimal scheduling\footnote{The existence of optimal policies has been shown only for identical machines~\cite{MRW84}. We believe that this can be generalized to unrelated machines and use the term ``optimal policy'' also in this setting. Formally completely correct would be to compare to the infimum of the expected objective value over all scheduling policies.} and the optimal fixed-assignment policy for the entire instance. Finally, on each machine the jobs are scheduled according to a single-machine policy, for which the resulting expected objective value can be bounded in terms of the surrogate cost. The choice of a good surrogate cost function and single-machine policy is crucial in this approach. They should allow for a strong analysis of both the greedy assignment and the single-machine rule. The best known policies for a single machine are the identical-machine policies by \citeauthor{Sch08}, described above, applied to a single machine. They were analyzed by comparing their expected outcomes to the sum of weighted mean busy times plus half processing times of the preemptive WSPT schedule for the deterministic counterparts (see \cref{sec:RSOS DSOS}). So by using this quantity as surrogate cost, we can simply adopt the analysis of \citeauthor{Sch08} for our single-machine rule. Our main insight is that this surrogate cost function is also perfectly suited for the analysis of the greedy assignment rule.

The enhancements for a known upper bound~$\Delta$ on the squared coefficients of variation are obtained by improving \citeauthor{Sch08}'s single-machine policies in this case. This is achieved by generalizing the techniques of \citet{GQS+02} to stochastic scheduling. Similar techniques have been applied to several other scheduling problems; see \cite{Sku06, CW09, Sku16}. Our policies coincide with \Citeauthor{GQS+02}'s online algorithms when $\Delta = 0$ (deterministic processing times), and our performance guarantees continuously generalize their guarantees to all $\Delta \ge 0$. For each such $\Delta$ these are the best known performance guarantees of any online policy as well as of any efficient policy for single-machine scheduling instances with squared coefficients of variation bounded by $\Delta$.

If the processing times are $\delta$-NBUE, then the squared coefficients of variation are bounded by $2\delta - 1$~\cite{MUV06}, so that we obtain linear bounds in terms of the parameter~$\delta$ as well. In the case of ordinary NBUE processing times ($\delta = 1$), we can bound the coefficients of variation only by $\Delta = 1$, for which all our policies outperform the previously known online policy, see \cref{fig:performance unrelated}. Yet the choice of the $\alpha_j$ in the single-machine rule can also be adapted directly to the parameter $\delta$, resulting in better bounds than those obtained by adapting the $\alpha_j$ to $\Delta = 2 \delta - 1$.

\subsection{Outline}

\Citeauthor{Sch08}'s RSOS and DSOS policy are applied to a single machine in \cref{sec:RSOS DSOS}. Some concepts introduced therein will be reused in all subsequent \lcnamecrefs{sec:unrelated}. In \cref{sec:unrelated}, the greedy assignment rule for unrelated machines is developed and its combination with the single-machine policies from the preceding \lcnamecref{sec:RSOS DSOS} is analyzed. This yields the main results of this paper, illustrated by the blue and orange lines in \cref{fig:performance unrelated}. The green and red curves are obtained by enhancements of the $\alpha_j$-scheduling rules applied on each machine when the coefficients of variation are bounded. These will be discussed in \cref{seq:refined policies}. Finally, some relations and comparisons to existing research are pointed out in \cref{sec:conclusion}. To keep the main part from becoming too long, we defer some technical proofs to \cref{apx:technical} and explain necessary adaptions to the analysis for $\delta$-NBUE processing times in \cref{apx:refined delta}.

\section{\texorpdfstring{\Citeauthor{Sch08}}{Schulz}'s RSOS and DSOS policy on a single machine} \label{sec:RSOS DSOS}

\Citet{Sch08} analyzed the RSOS and the DSOS policy for an arbitrary number of identical parallel machines and did not bother to refine the analysis in the special case of a single machine. By tailoring Schulz's analysis to this case, one can establish constant competitive ratios. Since this fact is not mentioned explicitly in the literature and does not seem to have become part of the folklore (see e.g.~\cite{Vre12}), we will state the two bounds in this \lcnamecref{sec:RSOS DSOS}. They arise as special cases of the more general analyses conducted in \cref{seq:refined policies}. Nevertheless, since the general analysis is quite technical and \citeauthor{Sch08}'s proof is not very detailed in parts, we present a detailed and simple proof for the bound for the RSOS policy. This leads on the one hand, together with \cref{sec:unrelated}, to a self-contained elegant proof  for one variant of our main result that exposes the main ideas. On the other hand, it serves as a preparation for the more technical analyses in \cref{seq:refined policies}.

We consider $n$ jobs~$j$ with processing times $\bm p_j$ to be scheduled on a single machine. Let $\mathrm{S}$ be a preemptive schedule of the deterministic jobs with weights~$w_j$, release dates~$r_j$, and processing times~$\bar p_j \coloneqq \E[\bm p_j]$, $j \in [n]$. For every $j \in [n]$ and $t \ge 0$ let
\[I_j^{\mathrm{S}}(t) \coloneqq \begin{cases*} 1 &if job~$j$ is being processed at time~$t$  in $\mathrm{S}$;\\ 0 &else. \end{cases*}\]
We assume that the functions~$I_j^{\mathrm S}$ are left-continuous. The \emph{mean busy time} of a job~$j$ in $\mathrm{S}$ is \[M_j^{\mathrm{S}} \coloneqq \frac{1}{\bar p_j} \cdot \int_0^\infty I_j^{\mathrm{S}}(t) \cdot t\dif t.\] Obviously, this is a lower bound on the completion time of $j$.

Let $\mathrm{S^p}$ be the schedule obtained by applying the \emph{preemptive weighted shortest processing time} (\emph{preemptive WSPT}) rule to the deterministic jobs. Assuming that $w_1/\bar p_1 \ge \cdots \ge w_n/\bar p_n$, whenever some jobs are available, the available job with smallest index is being processed. Note that the processed job always has maximum ratio of weight over its \emph{total} processing time and not over its current remaining processing time. \Citet{Goe96, Goe97} proved that for every preemptive schedule~$\mathrm S$ the mean busy time vector~$M^{\mathrm S}$ is a feasible solution to the following linear program:
{\renewcommand{\arraystretch}{1.25}
 \[\begin{array}{rr>{\displaystyle}l}
  (\mathrm{LP}) & \operatorname{minimize} &\sum_{j=1}^n w_j M_j\\
  &\text{subject to} &\sum_{j \in S} \bar p_j M_j \ge \sum_{j \in S} \bar p_j \cdot \left(r_{\min}(S) + \frac 1 2 \sum_{j \in S} \bar p_j\right)\quad \forall S \subseteq [n]
 \end{array}\]
}%
Moreover, the vector~$M^{\mathrm{S^p}}$ derived from the preemptive WSPT schedule is an optimal solution. Here, $r_{\min}(S) \coloneqq \min_{j \in S} r_j$ for $S \subseteq [n]$. Let $\Pi$ now be a non-preemptive scheduling policy for the stochastic jobs. The \emph{mean busy time} of a job~$j \in [n]$ under $\Pi$ is simply
\[\bm M_j^\Pi \coloneqq \bm C_j^{\Pi} - \frac{\bm p_j}{2},\]
where $\bm C_j^\Pi$ denotes the completion time of $j$ under $\Pi$. The word \textit{mean} here refers to the processing interval in each realization and does not denote an expected value.

\begin{lemma}\label{lem:lower bound}
 For every non-preemptive scheduling policy~$\Pi$ it holds that
 \[\sum_{j=1}^n w_j M_j^{\mathrm{S^p}} \le \E\left[\sum_{j=1}^n w_j \bm M_j^{\Pi}\right].\]
\end{lemma}
The \lcnamecref{lem:lower bound} follows from the fact that the vector~$(\E[\bm M_j^\Pi])_{j \in [n]}$ is a feasible solution to $(\mathrm{LP})$, which is a special case of \cite[Lemma~1]{Sch08}. The proof we provide here resembles the proof of \cite[Theorem~3.1]{MSU99}. 
\begin{proof}
 Let $\Pi$ be a non-preemptive scheduling policy. We show that the vector $(\E[\bm M_j^\Pi])_{j \in [n]}$ is a feasible solution to $(\mathrm{LP})$. As $M^{\mathrm{S^p}}$ is an optimal solution, this implies that \[\sum_{j=1}^n w_j M_j^{\mathrm{S^p}} \le \sum_{j=1}^n w_j \E[\bm M_j^\Pi] = \E\left[\sum_{j=1}^n w_j \bm M_j^{\Pi}\right].\]
 Let $S \subseteq [n]$. Then, as each schedule provides a feasible solution to $(\mathrm{LP})$, we have realizationwise
 \[\sum_{j \in S} \bm p_j \bm M_j^\Pi \ge \sum_{j \in S} \bm p_j \cdot \left(r_{\mathrm{min}}(S) + \frac 1 2 \cdot \sum_{j \in S} \bm p_j\right).\]
 As $\Pi$ is non-preemptive, we can decompose each mean busy time into $\bm M_j^\Pi = \bm S_j^\Pi + 1/2 \cdot \bm p_j$, where $\bm S_j^\Pi$ denotes the start time of $j$. Therefore,
 \begin{align*}
  \sum_{j \in S} \bm p_j \bm S_j^\Pi &\ge \sum_{j \in S} \bm p_j \cdot \left(r_{\mathrm{min}}(S) + \frac 1 2 \cdot \sum_{j \in S} \bm p_j\right) - \frac 1 2 \cdot \sum_{j \in S} \bm p_j^2 \\
  &= \sum_{j \in S} \bm p_j \cdot r_{\min}(S) + \frac 1 2 \cdot \sum_{\substack{i,j \in S\\i\neq j}} \bm p_i \bm p_j.
 \end{align*}
 Taking the expectation results in 
 \begin{align*}\sum_{j \in S} \E[\bm p_j] \cdot \E[\bm S_j^\Pi] &\ge \sum_{j \in S} \E[\bm p_j] \cdot r_{\min}(S) + \frac 1 2 \cdot \sum_{\substack{i,j \in S\\i\neq j}} \E[\bm p_i] \cdot \E[\bm p_j] \\ &= \sum_{j \in S} \E[\bm p_j] \cdot \left(r_{\min}(S) + \frac 1 2 \cdot \sum_{j \in S} \E[\bm p_j]\right) - \frac 1 2 \cdot \sum_{j \in S} \E[\bm p_j]^2,\end{align*}
 where we used that $\bm p_i, \bm p_j$ are independent for $i \neq j$ and that $\bm p_j$ and $\bm S_j^\Pi$ are independent because $\Pi$ is non-anticipative. By adding $1/2 \cdot \sum_{j \in S} \E[\bm p_j]^2$, we obtain the result.
\end{proof}

The \lcnamecref{lem:lower bound} implies that $\sum_{j=1}^n w_j \cdot (M_j^{\mathrm{S^p}} + \bar p_j/2)$ is a lower bound on the expected total weighted completion time resulting from an optimal scheduling policy.

For a job~$j \in [n]$ and $\alpha \in (0, 1]$ the \emph{$\alpha$-point} $C_j(\alpha)$ of $j$ is the first point in time when an $\alpha$-fraction of its deterministic counterpart has been completed in 
$\mathrm{S^p}$, i.e.,
\[C_j(\alpha) \coloneqq \min\biggl\{t \ge 0 \biggm| \int_0^t I_j^{\mathrm{S^p}}(s)\dif s \ge \alpha \bar p_j\biggr\}.\]
Each function~$\alpha \mapsto C_j(\alpha)$ for $j \in [n]$ is piecewise affine linear, where each piece has slope~$\bar p_j$ and the image is the union of all time intervals during which $j$'s deterministic counterpart is being processed in $\mathrm{S^p}$. Therefore, applying the substitution~$t \coloneqq C_j(\alpha)$ yields \cite[eq.~(3.1)]{GQS+02}:
\begin{equation}
 \int_0^1 C_j(\alpha)\dif \alpha = \frac{1}{\bar p_j} \cdot \int_0^{\infty} I_j^{\mathrm{S^p}}(t) \cdot t\dif t = M_j^{\mathrm{S^p}}. \label{eq:mean busy time alpha point}
\end{equation}

The RSOS policy proceeds as follows: Whenever a job~$j$ is released, the policy draws $\bm \alpha_j \in (0, 1]$ uniformly at random. It continually virtually constructs the preemptive WSPT schedule for the deterministic counterparts of the released jobs. Whenever the $\bm \alpha_j$-point of some job~$j$ is reached, the job is appended to a FIFO queue for processing on the real machine.

\begin{proposition} \label{prop:RSOS}
 On a single machine the RSOS policy satisfies
 \[\E\biggl[\,\sum_{j=1}^n w_j \bm M_j^{\mathrm{RSOS}}\biggr] \le 2 \cdot \sum_{j=1}^n w_j M_j^{\mathrm{S^p}} \le 2 \cdot \E\biggl[\,\sum_{j=1}^n w_j \bm M_j^{\mathrm{OPT}}\biggr].\]
\end{proposition}
The \lcnamecref{prop:RSOS} implies the weaker inequality
\[\E\biggl[\,\sum_{j=1}^n w_j \bm C_j^{\mathrm{RSOS}}\biggr] \le 2 \cdot \sum_{j=1}^n w_j \Bigl(M_j^{\mathrm{S^p}} + \frac{\bar p_j}{2}\Bigr) \le 2 \cdot \E\biggl[\,\sum_{j=1}^n w_j \bm C_j^{\mathrm{OPT}}\biggr],\]
which can also be derived from the proof of \cite[Theorem~3]{Sch08} together with \cref{lem:lower bound} and which would be sufficient to deduce the results of this paper. Nevertheless, we prove here the stronger statement to make this \lcnamecref{sec:RSOS DSOS} self-contained and because it could be useful in other contexts.
\begin{proof}
 Let $j \in [n]$ be fixed. The mean busy time of $j$ can be bounded by
 \begin{equation}
  \bm M_j^{\mathrm{RSOS}} \le C_j(\bm \alpha_j) + \sum_{k : C_k(\bm \alpha_k) < C_j(\bm \alpha_j)} \bm p_k + \frac{\bm p_j}{2}. \label{ineq:mean busy time RSOS}
 \end{equation}
 This is true because at time $C_j(\bm \alpha_j)$ all jobs~$k$ with $C_k(\bm \alpha_k) \le C_j(\bm \alpha_j)$ have been added to the FIFO queue, so that after this time there will be no idle time in the schedule constructed by the RSOS policy until the completion of $j$. Consider first a fixed $\alpha_j \in (0, 1]$. For every $k \in [n]$ let
 \begin{equation}
  \eta_k = \eta_k(j) \coloneqq \frac{1}{\bar p_k} \cdot \int_0^{C_j(\alpha_j)} I_k^{\mathrm{S^p}}(t)\dif t \label{eq:def eta_k}
 \end{equation}
 be the fraction of its processing time that the deterministic counterpart of job~$k$ receives before the $\alpha_j$-point of $j$. Thus, for $k \neq j$ the event~$C_k(\bm \alpha_k) < C_j(\alpha_j)$ occurs if and only if $\bm \alpha_k \le \eta_k$. The following inequality expresses the fact that the total processing that all jobs receive in $\mathrm{S^p}$ until time~$C_j(\alpha_j)$ is at most $C_j(\alpha_j)$:
 \begin{equation}
  \sum_{k \in [n] \setminus \{j\}} \eta_k \bar p_k + \alpha_j \bar p_j = \sum_{k=1}^n \eta_k \bar p_k \le C_j(\alpha_j). \label{eq:processing time before alpha point}
 \end{equation}
 We can bound the conditional expectation of $\bm M_j^{\mathrm{RSOS}}$ given $\bm \alpha_j = \alpha_j$ as follows: 
 \begin{align*}
  &\condexp{\bm M_j^{\mathrm{RSOS}}}{\bm \alpha_j = \alpha_j} \stackrel{\eqref{ineq:mean busy time RSOS}}\le C_j(\alpha_j) + \sum_{k \in [n] \setminus \{j\}} \Pr[C_k(\bm \alpha_k) < C_j(\alpha_j)] \cdot \E[\bm p_k] + \frac{\E[\bm p_j]}{2} \\
  &\quad = C_j(\alpha_j) + \sum_{k \in [n] \setminus \{j\}} \Pr[\bm \alpha_k \le \eta_k] \cdot \bar p_k + \frac{\bar p_j}{2} = C_j(\alpha_j) + \sum_{k \in [n] \setminus \{j\}} \eta_k \bar p_k + \frac{\bar p_j}{2} \\
  &\quad \stackrel{\mathclap{\eqref{eq:processing time before alpha point}}}\le 2 C_j(\alpha_j) + \Bigl(\frac 1 2 - \alpha_j\Bigr) \cdot \bar p_j.
 \end{align*}
 Here we used in the first inequality that $\bm \alpha_k$ and $\bm p_k$ are stochastically independent for all $k \in [n]$. Unconditioning yields
 \[\E[\bm M_j^{\mathrm{RSOS}}] = \int_0^1 \condexp{\bm M_j^{\mathrm{RSOS}}}{\bm \alpha_j = \alpha_j} \dif \alpha_j \le 2 \cdot \int_0^1 C_j(\alpha_j)\dif \alpha_j \stackrel{\eqref{eq:mean busy time alpha point}}= 2 M_j^{\mathrm{S^p}}.\]
 Hence, by linearity of expectation and \cref{lem:lower bound},
 \[\E\biggl[\,\sum_{j=1}^n w_j \bm M_j^{\mathrm{RSOS}}\biggr] \le 2 \cdot \sum_{j=1}^n w_j M_j^{\mathrm{S^p}} \le 2 \cdot \E\biggl[\,\sum_{j=1}^n w_j \bm M_j^{\mathrm{OPT}}\biggr].\qedhere\]
\end{proof}

Recall that $\phi = \frac{1+\sqrt 5}{2}$ denotes the golden ratio. The DSOS policy appends a job to the FIFO queue for processing on the real machine, as soon as it reaches its $(\phi-1)$-point. The following \lcnamecref{prop:DSOS} follows from more general results provided in \cref{subsec:refined DSOS}.
\begin{proposition} \label{prop:DSOS}
 On a single machine the DSOS policy satisfies
 \[\E\biggl[\,\sum_{j=1}^n w_j \bm M_j^{\mathrm{DSOS}}\biggr] \le (\phi+1) \cdot \sum_{j=1}^n w_j M_j^{\mathrm{S^p}} \le (\phi+1) \cdot \E\biggl[\,\sum_{j=1}^n w_j \bm M_j^{\mathrm{OPT}}\biggr].\]
\end{proposition}
As above, this \lcnamecref{prop:DSOS} implies that
\[\E\biggl[\,\sum_{j=1}^n w_j \bm C_j^{\mathrm{DSOS}}\biggr] \le (\phi+1) \cdot \sum_{j=1}^n w_j \Bigl(M_j^{\mathrm{S^p}} + \frac{\bar p_j}{2}\Bigr) \le (\phi+1) \cdot \E\biggl[\,\sum_{j=1}^n w_j \bm C_j^{\mathrm{OPT}}\biggr].\]

\section{Scheduling on unrelated machines} \label{sec:unrelated}

In this \lcnamecref{sec:unrelated} we describe and analyze stochastic online scheduling policies for unrelated machines. As outlined in the \lcnameref{sec:introduction}, these result as a combination of a greedy assignment rule and a single-machine $\alpha_j$-point scheduling rule. Our main results are obtained by using the RSOS and DSOS policy, dealt with in \cref{sec:RSOS DSOS}, as single-machine scheduling rules. The enhanced results with a priori information about the occurring processing time distributions will be discussed in \cref{seq:refined policies}.

The deterministic counterpart of a job~$j \in [n]$ has processing time~$\bar p_{ij} \coloneqq \E[\bm p_{ij}]$ when assigned to machine~$i \in [m]$. We assume that all $\bar p_{ij} > 0$. Whenever some job~$j$ is released, it is assigned to the machine~$i$ minimizing the immediate increase of the value \[\sum_{k \text{ assigned to } i} w_k \cdot \Bigl(M_k^{\mathrm S^{\mathrm p}_i} + \frac{\bar p_{ik}}{2}\Bigr).\] Here $\mathrm S_i^{\mathrm p}$ is the virtual preemptive WSPT schedule for the deterministic counterparts of the jobs assigned to machine~$i$. The increase incurred by assigning $j$ to $i$ is denoted by $\cost(j \to i)$. For every job~$j$ let $i^{\mathrm{GA}}(j)$ be the machine to which $j$ is assigned by this rule, and let $\mathrm{GA\text-S^p}$ be the combined virtual schedule on all machines.

The stochastic scheduling policies that use this algorithm to assign jobs to machines and schedule the jobs assigned to each machine according to the RSOS or the DSOS policy will be referred to as the \emph{greedy-assignment RSOS} (GA-RSOS) \emph{policy} and the \emph{greedy-assignment DSOS} (GA-DSOS) \emph{policy}, respectively. Since the assignment is based only on the virtual schedule and not on any observed processing times, these are fixed-assignment policies. In the remainder of this \lcnamecref{sec:unrelated} we will prove the following \lcnamecref{thm:unrelated}.

\begin{theorem} \label{thm:unrelated}
 \begin{enumerate}
  \item For instances with squared coefficients of variation bounded by $\Delta$ the GA-RSOS policy is $(8+4\Delta)$-competitive within the class of all scheduling policies and $8$-competitive within the class of fixed-assignment policies.
  \item For instances with squared coefficients of variation bounded by $\Delta$ the GA-DSOS policy is $((3+\sqrt 5)(2+\Delta))$-competitive within the class of all scheduling policies and $(2(3+\sqrt 5))$-competitive within the class of fixed-assignment policies.
 \end{enumerate}
\end{theorem}

Note that, in contrast to the problem without release dates, finding the best fixed-assignment policy is not equivalent to the deterministic scheduling problem. By applying \cref{prop:RSOS,prop:DSOS} to each machine, the expected total weighted completion times resulting from the two policies can be bounded by
\begin{align}
 \E\biggl[\,\sum_{j=1}^n w_j \cdot \bm C_j^{\mathrm{GA\text-RSOS}}\biggr] &\le 2 \cdot \sum_{j=1}^n w_j \cdot \Bigl(M_j^{\mathrm{GA\text-S^p}} + \frac{\bar p_{i^{\mathrm{GA}}(j),j}}{2}\Bigr), \label{ineq:upper bound GA-RSOS}\\
 \E\biggl[\,\sum_{j=1}^n w_j \cdot \bm C_j^{\mathrm{GA\text-DSOS}}\biggr] &\le (\phi+1) \cdot \sum_{j=1}^n w_j \cdot \Bigl(M_j^{\mathrm{GA\text-S^p}} + \frac{\bar p_{i^{\mathrm{GA}}(j),j}}{2}\Bigr). \label{ineq:upper bound GA-DSOS}
\end{align}
Thus, we only have to bound the quantity on the right-hand side. For this purpose we will use as another intermediate bound the optimal objective value~$\mathrm{OPT}(\mathrm{LP_r})$ of the following linear program. Assume that all expected processing times and release dates are even positive integers (this can be achieved by scaling), and let $T \coloneqq \max_{i \in [m]} \bigl(\max_{j \in [n]} r_j + \sum_{j=1}^n p_{ij} \bigr)$.
{\renewcommand{\arraystretch}{1.25}
 \[\begin{array}{r>{\displaystyle}l>{\qquad\forall}l}
  (\mathrm{LP_R}) \quad \min & \multicolumn{2}{l}{\displaystyle\sum_{j=1}^n w_j \cdot \sum_{i=1}^m \sum_{t=r_j}^{T-1} \left(\frac{y_{ijt}}{\bar p_{ij}} \Bigl(t + \frac 1 2\Bigr) + \frac{y_{ijt}}{2}\right)} \\
  \text{subject to}   &\sum_{i=1}^m \sum_{t=r_j}^{T-1} \frac{y_{ijt}}{\bar p_{ij}} = 1 &j \in [n] \\ 
  & \sum_{\substack{j \in [n]\\ t \ge r_j}} y_{ijt} \le 1 &i \in [m],\ t \in \{0,\dotsc,T-1\} \\
  &y_{ijt} \ge 0 & i \in [m],\ j \in [n],\ t \in \{0,\dotsc,T-1\}
 \end{array}\]%
}%
\Citeauthor{GMUX20} showed in \cite[inequality~(13)]{GMUX20} that for squared coefficients of variation bounded by $\Delta$ the optimal objective value of $(\mathrm{LP_R})$ is at most $1+\frac \Delta 2$ times the expected total weighted completion time of an optimal scheduling policy~$\Pi$. Moreover, it is well-known~\cite{SS02, AGK12} that $(\mathrm{LP_R})$ is a relaxation of the problem of scheduling the deterministic counterparts with preemption but without migration. If $\mathrm{FA}$ is an optimal fixed-assignment policy, then using the same assignment for the deterministic counterparts and applying the preemptive WSPT rule on each machine induces a feasible solution to $(\mathrm{LP_R})$.  By applying \cref{lem:lower bound} to each machine, we conclude that this yields a lower bound for the expected objective under $\mathrm{FA}$. Hence,
\begin{equation}
 \mathrm{OPT}(\mathrm{LP_R}) \le \min\Biggl\{\Bigl(1+\frac \Delta 2 \Bigr) \cdot \E\biggl[\,\sum_{j=1}^n w_j \bm C_j^{\Pi}\biggr],\ \E\biggl[\,\sum_{j=1}^n w_j \bm C_j^{\mathrm{FA}}\biggr]\Biggr\}. \label{ineq:bound LP_R}
\end{equation}
With these preparations, the only thing to be done for the proof of \cref{thm:unrelated} is to show that $\sum_{j=1}^n w_j \cdot \bigl(M_j^{\mathrm{GA\text-S^p}}+\frac{\bar p_{i^{\mathrm{GA}}(j),j}}{2}\bigr)$ is at most $4$ times the minimum objective value of $(\mathrm{LP_R})$. This will be done after some preparation in \cref{lem:dual fitting}, utilizing the technique of dual fitting.

For every $i \in [m]$ and $j \in [n]$ let $J_i(j)$ be the set of all jobs that have already been assigned to machine~$i$ when job~$j$ is being considered\footnote{If multiple jobs are released at time~$r_j$, they are considered one after the other in an arbitrary order.}. Finally, for every $j \in [n]$ and $t \ge r_j$ let 
\[\iota_j(t) \coloneqq 1 - \frac{1}{\bar p_{i^{\mathrm{GA}}(j),j}} \cdot \int_0^{t} I_j^{\mathrm{GA\text-S^p}}(s)\dif s\]
be the fraction of its processing time that the deterministic counterpart of job~$j$ receives in the virtual schedule after time~$t$.

\begin{lemma} \label{lem:cost assignment}
 For every $j \in [n]$ and $i \in [m]$ it holds that
 \[\cost(j \to i) = w_j \cdot \biggl(r_j + \bar p_{ij} + \sum_{\substack{k \in J_i(j) \\ \frac{w_k}{\bar p_{ik}} \ge \frac{w_j}{\bar p_{ij}}}} \iota_k(r_j) \cdot \bar p_{ik}\biggr) + \sum_{\substack{k \in J_i(j) \\ \frac{w_k}{\bar p_{ik}} < \frac{w_j}{\bar p_{ij}}}} w_k \cdot \iota_k(r_j) \cdot \bar p_{ij}.\]
\end{lemma}
\begin{proof}
 Let $j \in [n]$ and $i \in [m]$. In the preemptive WSPT schedule of the deterministic counterparts of all jobs in $J_i(j)$ no job is preempted after time~$r_j$, so from this time all jobs are scheduled in WSPT order. When now the deterministic counterpart of job~$j$ is added to this schedule, it is slated for the first point in time after $r_j$ when all jobs from $J_i(j)$ with larger ratio of weight over processing time will have been completed in this schedule. This is at time \[r_j + \sum_{\substack{k \in J_i(j) \\ \frac{w_k}{\bar p_{ik}} \ge \frac{w_j}{\bar p_{ij}}}} \iota_k(r_j) \cdot \bar p_{ik}.\]
 According to current knowledge, the job will not be interrupted, so its midpoint will be $\frac{\bar p_{ij}}{2}$ time units later. For all deterministic counterparts of jobs~$k \in J_i(j)$ with smaller ratio of weight over processing time, the remaining $\iota_k(r_j) \cdot \bar p_{ik}$ units of their processing time will be delayed by $\bar p_{ij}$. This increases their midpoints by $\iota_k(r_j) \cdot \bar p_{ij}$.
\end{proof}

\begin{lemma} \label{lem:dual fitting}
 \[\sum_{j=1}^n w_j \cdot \Bigl(M_j^{\mathrm{GA\text-S^p}}+\frac{\bar p_{i^{\mathrm{GA}}(j),j}}{2}\Bigr) \le 4 \cdot \mathrm{OPT}(\mathrm{LP_r}).\]
\end{lemma}
\begin{proof}
The following linear program is the dual of $(\mathrm{LP_y^{det}})$:
{\renewcommand{\arraystretch}{1.25}
 \[\begin{array}{r>{\displaystyle}l>{\forall}l}
  (\mathrm{D})\quad \max &\sum_{j=1}^n \chi_j - \sum_{i=1}^m \sum_{t=0}^{T-1} \psi_{it} \\
  \text{subject to} & \frac{\chi_j}{\bar p_{ij}} \le \psi_{it} + w_j \cdot \Bigl(\frac{t+1/2}{\bar p_{ij}} + \frac 1 2\Bigr) &i \in [m],\ j \in [n],\ t \in \{0,\dotsc,T-1\} \\
  &\psi_{it} \ge 0 &i \in [m],\ t \in \{0,\dotsc,T-1\}
 \end{array}\]
}%

We define a dual solution as follows:
\begin{align*}
 \chi_j &\coloneqq \frac 1 2 \cdot \cost\bigl(j \to i^{\mathrm{GA}}(j)\bigr) &&\text{for all } j \in [n],\\
 \psi_{it} &\coloneqq \frac 1 2 \cdot \sum_{k : i(k) = i} \iota_k(2t) \cdot w_k &&\text{for all } i \in [m],\ t \in \{0,\dotsc,T-1\}.
\end{align*}
We now show that this solution is feasible, i.e., that
\[\cost\bigl(j \to i^{\mathrm{GA}}(j)\bigr) \le \sum_{k : i(k)=i} \iota_k(2t) \cdot w_k \bar p_{ij} + w_j \cdot (2t+1 + \bar p_{ij})\]
for all $i \in [m]$, $j \in [n]$, and $t \in \{r_j,\dotsc,T-1\}$. Using \cref{lem:cost assignment}, we deduce
\allowdisplaybreaks%
\begin{align*}
 &\cost\bigl(j \to i^{\mathrm{GA}}(j)\bigr) \le \cost(j \to i) \\
 ={} &w_j \cdot \biggl(r_j + \bar p_{ij} + \sum_{\substack{k \in J_i(j) \\ \frac{w_k}{\bar p_{ik}} \ge \frac{w_j}{\bar p_{ij}}}} \iota_k(r_j) \cdot \bar p_{ik}\biggr) + \sum_{\substack{k \in J_i(j) \\ \frac{w_k}{\bar p_{ik}} < \frac{w_j}{\bar p_{ij}}}} w_k \cdot \iota_k(r_j) \cdot \bar p_{ij} \\
 ={} &w_j \cdot (r_j + \bar p_{ij}) + \sum_{k \in J_i(j)} \iota_k(r_j) \cdot \min\{w_j \bar p_{ik},\ w_k \bar p_{ij}\} \\
 \le{} &w_j \cdot (r_j + \bar p_{ij}) + \sum_{k \in J_i(j)} \bigl((\iota_k(r_j) - \iota_k(2t)) \cdot w_j \bar p_{ik} + \iota_k(2t) \cdot w_k \bar p_{ij}\bigr) \\
 \le{} &w_j \cdot (r_j + \bar p_{ij}) + w_j \cdot \sum_{k \in J_i(j)} \bigl(\iota_k(r_j)-\iota_k(2t)\bigr) \cdot \bar p_{ik} + \sum_{k : i(k)=i} \iota_k(2t) \cdot w_k \bar p_{ij} \\
 \le{} &w_j \cdot (2t + \bar p_{ij}) + \sum_{k : i(k)=i} \iota_k(2t) \cdot w_k \bar p_{ij},
\end{align*}
\allowdisplaybreaks[0]%
where we used in the last inequality that the total processing time that jobs assigned to machine~$i$ receive between $r_j$ and $2t$ is at most $2t-r_j$. Furthermore,
\begin{align}
 \sum_{j=1}^n \chi_j &= \frac 1 2 \cdot \sum_{j=1}^n \cost\bigl(j \to i^{\mathrm{GA}}(j)\bigr) = \frac 1 2 \cdot \sum_{j=1}^n w_j \cdot \Bigl( M_j^{\mathrm{GA\text-S^p}} + \frac{\bar p_{i^{\mathrm{GA}}(j),j}}{2} \Bigr), \label{eq:sum chi_j}\\
 \sum_{i=1}^m \sum_{t=0}^{T-1} \psi_{it} &= \frac 1 2 \cdot \sum_{i=1}^m \sum_{t=0}^{T-1} \sum_{j : i^{\mathrm{GA}}(j)=i} \iota_j(2t) \cdot w_j \le \frac 1 4 \cdot \sum_{j=1}^n w_j \cdot \Bigl(M_j^{\mathrm{GA\text-S^p}} + \frac{\bar p_{i^{\mathrm{GA}}(j),j}}{2} \Bigr) \label{ineq:sum psi_it}.
\end{align}
The first equation clearly holds because the sum of the cost increments in all steps of the greedy assignment rule equals the resulting total cost. The intuition for the second inequality is that the sum of the $\iota_j(2t)$ is roughly half the integral over the time of the remaining unfinished weight, when the processed part of each job is already regarded as completed (or equivalently, jobs are considered to be infinitesimally small). This integral is exactly the sum of weighted mean busy times. It is true that we overestimate it by its upper sum, but this can be offset against half the total weighted processing time in the virtual schedule. The formal calculation is as follows:
\allowdisplaybreaks%
\begin{align*}
 &\sum_{i=1}^m \sum_{t=0}^{T-1} \sum_{j : i^{\mathrm{GA}}(j) = i} \iota_j(2t) \cdot w_j \\
 ={} &\sum_{i=1}^m \sum_{j : i^{\mathrm{GA}}(j) = i} w_j \cdot \sum_{t=0}^{T-1} \biggl(1-\frac{1}{\bar p_{ij}} \cdot \int_{0}^{2t} I_j^{\mathrm{GA\text-S^p}}(s)\dif s\biggr)\\
 ={} &\sum_{i=1}^m \sum_{j : i^{\mathrm{GA}}(j) = i} \frac{w_j}{\bar p_{ij}} \cdot \sum_{t=0}^{T-1} \int_{0}^{\infty} \mathds 1_{(2t, \infty)}(s) \cdot I_j^{\mathrm{GA\text-S^p}}(s)\dif s \\
 ={} &\sum_{i=1}^m \sum_{j : i^{\mathrm{GA}}(j) = i} \frac{w_j}{\bar p_{ij}} \cdot \int_{0}^{\infty} \left\lceil\frac s 2\right\rceil \cdot I_j^{\mathrm{GA\text-S^p}}(s)\dif s \\
 ={} &\sum_{i=1}^m \sum_{j : i^{\mathrm{GA}}(j) = i} \frac{w_j}{\bar p_{ij}} \cdot \biggl(\int_{0}^{\infty} \frac s 2 \cdot I_j^{\mathrm{GA\text-S^p}}(s)\dif s + \int_{0}^{\infty} \left(\left\lceil\frac s 2\right\rceil - \frac s 2\right) \cdot I_j^{\mathrm{GA\text-S^p}}(s)\dif s\biggr) \\
 ={} &\sum_{i=1}^m \sum_{j : i^{\mathrm{GA}}(j) = i} w_j \biggl(\frac{M_j^{\mathrm{GA\text-S^p}}}{2} + \frac{1}{\bar p_{ij}} \sum_{t : I_j^{\mathrm{GA\text-S^p}}(s)=1\, \forall s \in (2t, 2(t+1)]} \underbrace{\int_{2t}^{2(t+1)} t+1-\frac s 2\dif s}_{=1}\biggr) \\
 ={} &\frac 1 2 \cdot \sum_{i=1}^m \sum_{j : i^{\mathrm{GA}}(j) = i} w_j \cdot (M_j^{\mathrm{GA\text-S^p}}+1) \le \frac 1 2 \cdot \sum_{j=1}^n w_j \cdot \Bigl(M_j^{\mathrm{GA\text-S^p}} + \frac{\bar p_{i^{\mathrm{GA}}(j),j}}{2}\Bigr),
\end{align*}
where we used that all release dates and expected processing times are even positive integers and thus in $\mathrm{GA\text-S^p}$ starts, completions, and preemptions only take place at times in $[0, T] \cap 2 \Z$. \Cref{eq:sum chi_j} and \ineq{ineq:sum psi_it} imply that
\[\sum_{j=1}^n \chi_j - \sum_{i=1}^m \sum_{t=0}^{T-1} \psi_{it} \ge \frac 1 4 \cdot \sum_{j=1}^n w_j \cdot \Bigl(M_j^{\mathrm{GA\text-S^p}} + \frac{\bar p_{i^{\mathrm{GA}}(j),j}}{2} \Bigr).\]
So the claim follows from weak linear programming duality.
\end{proof}

Now we put everything together in order to prove \cref{thm:unrelated}.

\begin{proof}[\noindent Proof of \cref{thm:unrelated}]
 \begin{enumerate}
  \item Let $\Pi$ be an arbitrary scheduling policy, and let $\mathrm{FA}$ be an arbitrary fixed-assignment policy. Then
  \begin{multline*}
   \E\biggl[\,\sum_{j=1}^n w_j \cdot \bm C_j^{\mathrm{GA\text-RSOS}}\biggr] \stackrel{\eqref{ineq:upper bound GA-RSOS}}\le 2 \cdot \sum_{j=1}^n w_j \cdot \Bigl(M_j^{\mathrm{GA\text-S^p}} + \frac{\bar p_{i^{\mathrm{GA}}(j),j}}{2}\Bigr) \\
   \stackrel{\text{Lem.~\ref{lem:dual fitting}}}\le 8 \cdot \mathrm{OPT}(\mathrm{LP_R}) \\
   \stackrel{\eqref{ineq:bound LP_R}}\le \min\Biggl\{(8 + 4 \Delta) \cdot \E\biggl[\,\sum_{j=1}^n w_j \bm C_j^{\Pi}\biggr],\ 8 \cdot \E\biggl[\,\sum_{j=1}^n w_j \bm C_j^{\mathrm{FA}}\biggr]\Biggr\}.
  \end{multline*}
  \item The proof for the DSOS policy works exactly in the same with \ineq{ineq:upper bound GA-RSOS} replaced by \ineq{ineq:upper bound GA-DSOS}. \qedhere
 \end{enumerate}
\end{proof}

\section{Refined scheduling policies} \label{seq:refined policies}

In this \lcnamecref{seq:refined policies} we will describe and analyze the refined policies that are aware of an upper bound~$\Delta$ for the coefficients of variation of the random processing times. When not stated explicitly, we again consider a single machine and use the notation and assumptions from \cref{sec:RSOS DSOS}. In particular, we always assume that $w_1/\bar p_1 \ge \cdots \ge w_n / \bar p_n$. The resulting bounds are illustrated in \cref{fig:performance single}.
\begin{figure}
 \centering
 \begin{tikzpicture}[font=\footnotesize]
  \begin{axis}[
    xmin=0, xmax=2.05, ymin=1, ymax=2.8, xlabel={$\Delta$}, ylabel={performance guarantee},
    ytick={1, 1.5, 1.6853, 2, {1+sqrt(2)}, {(3+sqrt(5))/2}},
    yticklabels={$1$, $1.5$, $1.6853$, $2$, $1+\sqrt 2$, $\phi+1$}
   ]
   \addplot[myblue, very thick] coordinates {(0, {(3+sqrt(5))/2}) (2, {(3+sqrt(5))/2})} node[yshift=2pt, right] {deterministic (\Cref{prop:DSOS})};
   \addplot[mygreen, very thick, domain=0:1] (x, {1+(sqrt(x)+sqrt(32-8*sqrt(x)+x))/4});
   \addplot[mygreen, very thick, domain=1:2] plot (x, {1+(2*(2+x))/(sqrt(8+12*x+5*x*x)-x)}) node[anchor=north west, align=left, yshift=7pt] {deterministic, optimized for $\Delta$\\(\Cref{cor:SOS alpha Delta})};
   \addplot[myorange, very thick] coordinates {(0, 2) (2, 2)} node [yshift=2pt, right] {randomized (\Cref{prop:RSOS})};
   \addplot[very thick, myred] table[x expr=(\coordindex+1)/1000, y index=0] {rand_single_machine.dat} node[yshift=9pt, align=left, anchor=north west] {randomized, optimized for $\Delta$\\(\Cref{cor:RSOS density Delta})};
   \draw[mygreen] (0, {1+sqrt(2)}) node [circle, fill, minimum size=6pt, inner sep=0pt] {} node[below right] {\cite{GQS+02}};
   \draw[myred] (0, 1.6853) node [circle, fill, minimum size=6pt, inner sep=0pt] {} node[below right] {\cite{GQS+02}};
   \draw[dashed] (1, 1) -- +(0, 2.7);
  \end{axis}
 \end{tikzpicture}
 \caption{Performance guarantees of different online policies for a single machine as a function of upper bound~$\Delta$ for the squared coefficients of variation. The dashed line indicates the guarantees for exponentially distributed processing times.}
 \label{fig:performance single}
\end{figure}
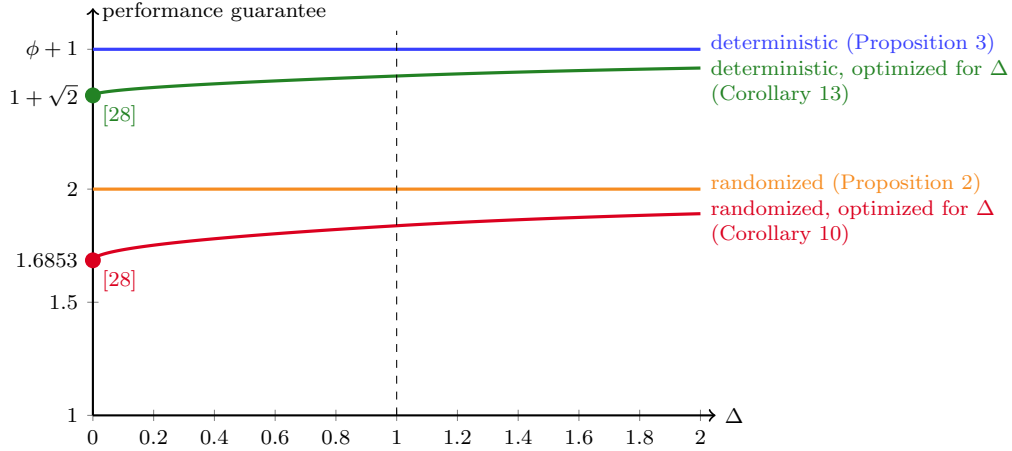
For example, it follows from these bounds that for exponentially distributed processing times, which have coefficient of variation $1$, there is a $1.839$-competitive randomized and a $2.5$-competitive deterministic online policy, improving significantly upon the general bounds of $2$ and $2.619$ from \cref{sec:RSOS DSOS}. A similar refinement is also possible when the processing times are $\delta$-NBUE. This leads to still better bounds of $1.783$ and $2.452$, respectively, for exponentially distributed processing times. In this \lcnamecref{seq:refined policies}, however, we will limit ourselves to the analysis for bounded coefficients of variation in order to keep its length within reason. We will present the modifications necessary for $\delta$-NBUE processing times in \cref{apx:refined delta}

As announced in the \lcnameref{sec:introduction}, the policies developed in this \lcnamecref{seq:refined policies} generalize the online algorithms for deterministic processing times by \citet{GQS+02} for. The analysis given here also extends their analysis and includes it as the special case~$\Delta = 0$.
We start in \cref{subsec:delayed list scheduling} with general delayed list scheduling in order of $\alpha_j$-scheduling, which is the basis of both the randomized and the deterministic single-machine policy. In \cref{subsec:refined RSOS,subsec:refined DSOS} we will then study these two types of policies separately.

\subsection{Delayed list scheduling in order of \texorpdfstring{$\alpha_j$}{alpha\_j}-points} \label{subsec:delayed list scheduling}

Assume that $\CV[\bm p_j]^2 \le \Delta$ for all jobs~$j \in [n]$. For every $A = (\alpha_j)_{j\in \N} \in (0, 1]^\N$ the $\mathrm{SOS}(A)$ policy schedules each job~$j$ as early as possible after time~$C_j(\alpha_j)$ and after all jobs~$k$ with $C_k(\alpha_k) < C_j(\alpha_j)$. For $\alpha \in (0, 1]$ we write $\mathrm{SOS}(\alpha)$ for $\mathrm{SOS}(\alpha,\alpha,\dotsc)$. Consider now a fixed $A \in (0, 1]^\N$ and $j \in [n]$, and for $k \in [n]$ let $\eta_k = \eta_k(j)$ be defined by \cref{eq:def eta_k} for this job~$j$. For $x \in \R$ the positive part is denoted by $x^+ \coloneqq \max\{x, 0\}$. The following lemma is a stochastic analogue of \cite[Lemma~3.1]{GQS+02}.

\begin{lemma} \label{lem:bound completion time SOS}
 \[\bm C_j^{\mathrm{SOS}(A)} \le C_j(\alpha_j) + \sum_{k : \alpha_k \le \eta_k} \bigl(\bm p_k-(\eta_k-\alpha_k) \cdot \bar p_k\bigr)^+.\]
\end{lemma}
\begin{proof}
 For $j \in [n]$ let $\bm S_j^{\mathrm{SOS}(A)}$ denote the start time of job~$j$ in the schedule produced by the policy~$\mathrm{SOS}(A)$. Let $\bm \ell$ be the first job in this schedule such that there is no idle time between $\bm S_{\bm \ell}^{\mathrm{SOS}(A)}$ and $\bm C_j^{\mathrm{SOS}(A)}$, and let $\bm K$ be the set of jobs processed during this interval in this schedule. Then $\bm S_{\bm \ell}^{\mathrm{SOS}(A)} = C_{\bm \ell}(\alpha_{\bm \ell})$ and $\bm C_j^{\mathrm{SOS}(A)} = C_{\bm \ell}(\alpha_{\bm \ell}) + \sum_{k \in \bm K} \bm p_k$. On the other hand, the $\alpha_j$-point of each job~$k \in \bm K$ lies between $C_{\bm \ell}(\alpha_{\bm \ell})$ and $C_j(\alpha_j)$, meaning that its deterministic counterpart is processed for at least $(\eta_k - \alpha_k) \cdot \bar p_k$ time units within the interval~$\bigl(C_{\bm \ell}(\alpha_{\bm \ell}), C_j(\alpha_j)\bigr]$ in $\mathrm{S^p}$. Thus, $C_j(\alpha_j) \ge C_{\bm \ell}(\alpha_{\bm \ell}) + \sum_{k \in \bm K} (\eta_k - \alpha_k) \cdot \bar p_k$. Hence,
 \begin{multline*}
  \bm C_j^{\mathrm{SOS}(A)} \le C_j(\alpha_j) + \sum_{k \in \bm K} \bigl(\bm p_k - (\eta_k - \alpha_k) \cdot \bar p_k\bigr) \\
  \le C_j(\alpha_j) + \sum_{k \in \bm K} \bigl(\bm p_k - (\eta_k - \alpha_k) \cdot \bar p_k\bigr)^+ \le C_j(\alpha_j) + \sum_{k : \alpha_k \le \eta_k} \bigl(\bm p_k - (\eta_k - \alpha_k) \cdot \bar p_k\bigr)^+
 \end{multline*}
 because every $k \in \bm K$ satisfies that $\alpha_k \le \eta_k$.
\end{proof}

\begin{lemma} \label{lem:parametric}
 Let $g$ be defined by \cref{eq:def g}. Then
\[\E[\bm C_j^{\mathrm{SOS}(A)}] \le C_j(\alpha_j) + \sum_{k : \alpha_k \le \eta_k} \big(1-g(\Delta) \cdot (\eta_k - \alpha_k)\bigr) \cdot \bar p_k.\]
\end{lemma}
\begin{proof}
 Taking the expectation of \cref{lem:bound completion time SOS} and applying \cref{lem:expected positive part Delta} from the appendix, which generalizes \cite[Lemma~A.2]{GMUX20}, yields
 \begin{align*}
  \E[\bm C_j^{\mathrm{SOS}(A)}] &\le C_j(\alpha_j) + \sum_{k: \alpha_k \le \eta_k} \E\bigl[(\bm p_k - (\eta_k - \alpha_k) \cdot \bar p_k)^+\bigr]\\
  &\le C_j(\alpha_j) + \sum_{k: \alpha_k \le \eta_k} \bigl(1-g(\Delta) \cdot (\eta_k-\alpha_k)\bigr) \cdot \bar p_k. \qedhere
 \end{align*}
\end{proof}

\subsection{Refined randomized online policies} \label{subsec:refined RSOS}

We now turn to refining Schulz's RSOS policy by choosing the $\bm \alpha_j$ according to more sophisticated probability distributions. For a given probability density function~$f$ on $(0, 1]$ the $\mathrm{RSOS}(f)$ policy acts as the RSOS policy with the exception that the $\bm \alpha_j$ are drawn independently according to the distribution given by $f$. The RSOS policy is the special case when $f = \mathds 1_{(0, 1]}$.

We start by recalling the analysis by \citeauthor{GQS+02} of the structure of the preemptive WSPT schedule~$\mathrm{S^p}$. We denote by $S_j^{\mathrm{S^p}}$ the start time of the deterministic counterpart of job~$j \in [n]$ in $\mathrm{S^p}$. Whenever a job~$j$ is preempted in $\mathrm{S^p}$, then all jobs processed between that moment and the resumption of $j$ have smaller index than $j$ and are therefore completely processed within this time. Hence, for a fixed job~$j$ the set~$[n] \setminus \{j\}$ can be partitioned into two subsets~$N_1 = N_1(j)$ and $N_2 = N_2(j)$. Here $N_1$ is the set of all jobs whose deterministic counterparts are completely processed outside the interval~$(S_j^{\mathrm{S^p}}, C_j^{\mathrm{S^p}}]$, and $N_2$ consists of all jobs $k \neq j$ whose deterministic counterparts are completely processed within $(S_j^{\mathrm{S^p}}, C_j^{\mathrm{S^p}}]$ in $\mathrm{S^p}$. For every $A = (\alpha_j)_{j \in \N} \in (0, 1]^\N$ and $k \in N_1$ the quantity~$\eta_k = \eta_k(j)$ indicates, independently of $\alpha_j$, the fraction of its processing time that the deterministic counterpart of $k$ receives in $\mathrm{S^p}$ before time~$S_j^{\mathrm{S^p}}$. For every $k \in N_2$ let 
\[\mu_k = \mu_k(j) \coloneqq \frac{1}{\bar p_j} \cdot \int_0^{S_k^{\mathrm{S^p}}} I_j^{\mathrm{S^p}}(t)\dif t = \max\bigl\{\alpha \in (0, 1) \bigm| C_j(\alpha) \le S_k^{\mathrm{S^p}}\bigr\}\]
be the fraction of its processing time that the deterministic counterpart of $j$ receives in $\mathrm{S^p}$ before time~$S_k^{\mathrm{S^p}}$. 
Then for $k \in N_2$
\[\eta_k = \begin{cases*} 0 &if $\alpha_j \le \mu_k$;\\ 1 &if $\alpha_j > \mu_k$ \end{cases*}\]
because when the deterministic counterpart of job~$k$ preempts that of job~$j$, this is because it has lower index, so that the counterpart of $j$ is resumed only after that of $k$ has been completed.
Using the introduced notation we can represent the $\alpha$-point of $j$ for any $\alpha \in (0, 1]$ as
\begin{equation}
 C_j(\alpha) = S_j^{\mathrm{S^p}} + \sum_{\substack{k \in N_2\\\alpha > \mu_k}} \bar p_k + \alpha \cdot \bar p_j.\label{eq:alpha point}
\end{equation}

\begin{theorem} \label{thm:RSOS Delta}
 Let $f$ be a probability density function on $(0, 1]$, and let $c \in \R$ be such that for all $x \in (0, 1]$ the following two conditions are fulfilled:
 \begin{enumerate}
  \item $\int_0^x \bigl(1-g(\Delta) \cdot (x-\alpha)\bigr) \cdot f(\alpha)\dif \alpha \le (c-1) \cdot x$,
  \item $\left(2-g(\Delta) \cdot \left(1-\int_0^1 \alpha \cdot f(\alpha)\dif \alpha\right)\right) \cdot \int_{1-x}^1 f(\alpha)\dif \alpha \le c \cdot x$.
 \end{enumerate}
 Then the $\mathrm{RSOS}(f)$ policy satisfies
 \begin{equation}\E\biggl[\,\sum_{j=1}^n w_j \cdot \bm C_j^{\mathrm{RSOS}(f)}\biggr] \le c \cdot \sum_{j=1}^n w_j \cdot \Bigl( M_j^{\mathrm{S^p}} + \frac{\bar p_j}{2}\Bigr) \le c \cdot \E\biggl[\,\sum_{j=1}^n w_j \cdot \bm C_j^{\mathrm{OPT}}\biggr].\label{ineq:RSOS(f)} \end{equation}
\end{theorem}
\begin{proof}
 This proof is inspired by the proof of \cite[Theorem~3.9]{GQS+02}. Let $j \in [n]$, and let $N_1$ and $N_2$ as well as $\eta_k$, $k \in [n]$, and $\mu_k$, $k \in N_2$, be defined as above with respect to this job~$j$. The total processing that jobs receive on $\mathrm{S^p}$ before time~$S_j^{\mathrm{S^p}}$ cannot exceed $S_j^{\mathrm{S^p}}$, i.e.,
 \begin{equation}\sum_{k \in N_1} \eta_k \cdot \bar p_k \le S_j^{\mathrm{S^p}}.\label{ineq:total processing time before S_j}\end{equation}
 Substituting \cref{eq:alpha point} into \cref{lem:parametric} results in
 \begin{equation}\begin{split}
  \E[\bm C_j^{\mathrm{SOS}(A)}] \le S_j^{\mathrm{S^p}} &+ \sum_{\substack{k \in N_1 \\ \alpha_k \le \eta_k}} \bigl(1-g(\Delta) (\eta_k - \alpha_k)\bigr) \bar p_k \\ &+ \sum_{\substack{k \in N_2 \\ \alpha_j > \mu_k}} \bigl(2-g(\Delta) (1 - \alpha_k)\bigr) \bar p_k + (1+\alpha_j) \bar p_j. \label{ineq:upper bound expected completion time SOS(A) Delta}
 \end{split}\end{equation}
 By the law of total expectation and the independence of the $\bm \alpha_j$, $j \in [n]$, and the $\bm p_j$, $j \in [n]$,
 \allowdisplaybreaks
 \begin{align*}
  &\E[\bm C_j^{\mathrm{RSOS}(f)}] = \int_0^1 \cdots \int_0^1 \E[\bm C_j^{\mathrm{SOS}(\alpha_1,\dotsc,\alpha_n)}] \cdot f(\alpha_1) \dotsm f(\alpha_n)\dif \alpha_1 \cdots \dif \alpha_n \\
  \stackrel{\mathclap{\eqref{ineq:upper bound expected completion time SOS(A) Delta}}}\le{} &S_j^{\mathrm{S^p}} + \sum_{k \in N_1} \int_0^{\eta_k} \bigl(1-g(\Delta) \cdot (\eta_k - \alpha_k)\bigr) \cdot f(\alpha_k)\dif \alpha_k \cdot \bar p_k \\ &\phantom{S_j^{\mathrm{S^p}}}{} + \sum_{k \in N_2} \Biggl(2-g(\Delta) \cdot \biggl(1-\int_0^1 \alpha_k \cdot f(\alpha_k)\dif \alpha_k\biggr)\Biggr) \cdot \int_{\mu_k}^1 f(\alpha_j)\dif \alpha_j \cdot \bar p_k \\ &\phantom{S_j^{\mathrm{S^p}}}{} + \biggl(1+\int_0^1 \alpha_j f(\alpha_j)\dif \alpha_j\biggr) \cdot \bar p_j \\
  \le{} &S_j^{\mathrm{S^p}} + \sum_{k \in N_1} (c-1) \eta_k \bar p_k + \sum_{k \in N_2} c (1-\mu_k) \bar p_k \\
  &\phantom{S_j^{\mathrm{S^p}}}{} + \Biggl(2-g(\Delta) \biggl(1-\int_0^1 \alpha_j f(\alpha_j)\dif \alpha_j\biggr)\Biggr) \cdot \bar p_j \\
  \stackrel{\mathclap{\eqref{ineq:total processing time before S_j}}}\le{} &c \cdot S_j^{\mathrm{S^p}} + c \cdot \sum_{k \in N_2} (1-\mu_k) \bar p_k + c \bar p_j \\
  ={} &c \cdot \biggl(\int_0^1 S_j^{\mathrm{S^p}} + \sum_{\substack{k \in N_2\\ \alpha > \mu_k}} \bar p_k + \alpha \bar p_j\dif \alpha + \frac{\bar p_j}{2}\biggr)\\
  \stackrel{\mathclap{\eqref{eq:alpha point}}}={} &c \cdot \biggl(\int_0^1 C_j(\alpha)\dif \alpha + \frac{\bar p_j}{2}\biggr) \stackrel{\eqref{eq:mean busy time alpha point}}= c \cdot \Bigl(M_j^{\mathrm{S^p}} + \frac{\bar p_j}{2}\Bigr).
 \end{align*}
 Consequently, by linearity of expectation and \cref{lem:lower bound},
 \[\E\biggl[\,\sum_{j=1}^n w_j \cdot \bm C_j^{\mathrm{RSOS}(f)}\biggr] \le c \cdot \sum_{j=1}^n w_j \cdot \Bigl(M_j^{\mathrm{S^p}} + \frac{\bar p_j}{2}\Bigr) \le c \cdot \biggl[\,\sum_{j=1}^n w_j \cdot \bm C_j^{\mathrm{OPT}}\biggr]. \qedhere\]
\end{proof}

For every $\Delta \ge 0$, when $f = \mathds 1_{(0, 1]}$, the two conditions are satisfied for $c=2$. Hence, \cref{thm:RSOS Delta} generalizes \cref{prop:RSOS}. For $\Delta \to \infty$, every other distribution gives a worse bound, so that for unbounded $\Delta$ the uniform distribution is the best choice and no better performance guarantee follows from the \lcnamecref{thm:RSOS Delta}. But also when $\Delta$ is taken into account, no better bound for uniformly distributed $\bm \alpha_j$ is implied. So in contrast to the greedy assignment rule, where we get an a posteriori bound in terms of the maximum squared coefficient of variation, here we can only get an improvement by explicitly adapting the policy to the known bound~$\Delta$. In the following \lcnamecref{cor:RSOS density Delta}, generalizing \cite[Theorem~3.9]{GQS+02}, we provide a good probability density function for each $\Delta$.

\begin{restatable}{corollary}{Gooddensity} \label{cor:RSOS density Delta}
 Let $D \coloneqq g(\Delta)$, let $\gamma$ be a solution to the equation
 \begin{multline*}
  \mathrm e^{D\gamma} \Bigl(\mathrm e^{D \gamma} \bigl(1+D(1-\gamma)\bigr) \bigl(D (\gamma -D(1-\gamma) + \ln(1+D (1-\gamma)))-1\bigr) + D \gamma (D-2) + 2\Bigr) \\ = 1-D \end{multline*}
 satisfying $0<\gamma<1$. Let $\theta \coloneqq \gamma + \frac{1}{D} \cdot \ln\bigl(1+D \cdot (1-\gamma)\bigr)$, and let $c \coloneqq 1+\frac{D}{\mathrm e^{D \theta}-1}$. Let further
 \[f_\Delta(\alpha) \coloneqq \begin{cases*} (c-1) \cdot \mathrm e^{D \cdot \alpha} &if $\alpha \le \theta$;\\ 0 &otherwise. \end{cases*}\]
 Then the $\mathrm{RSOS}(f_\Delta)$ policy is $c$-competitive for instances with squared coefficients of variation bounded by $\Delta$.
\end{restatable}
In the proof it is shown that the defined function is indeed a probability density function on $(0, 1]$ and that it satisfies the two conditions from \cref{thm:RSOS Delta}. Since it is very technical, it is moved to \cref{apx:technical}. Ignoring issues about the representation and transformation of real random variables to the prescribed distribution, the resulting policy can be executed efficiently. We can now also conclude the improved bound for unrelated machines illustrated in \cref{fig:performance unrelated}.

\begin{corollary}
 Let $c$ be as in \cref{cor:RSOS density Delta}. There is a randomized online policy for unrelated machines that is $(c (4 + 2\Delta))$-competitive within the class of all stochastic scheduling policies and $4c$ competitive within the class of fixed-assignment policies.
\end{corollary}

For deterministic processing times this yields a $6.742$-competitive randomized online algorithm for unrelated machines. To our knowledge no better online algorithm with immediate dispatch has been established for this problem.

\subsection{Refined deterministic online policies} \label{subsec:refined DSOS}

We consider the $\mathrm{SOS}(\alpha)$ policy for some $\alpha \in (0, 1]$. Schulz' DSOS policy is simply the policy~$\mathrm{SOS}(\phi-1)$. We begin by reviewing the concept of canonical sets introduced by \citet{Goe97}. Let $k \in [n]$, and let $\mathrm{S^p}|_{[k]}$ be the restriction of $\mathrm{S^p}$ to the first $k$ jobs, which is exactly the preemptive WSPT schedule of the deterministic counterparts of the first $k$ jobs. The family $\mathcal C(k)$ of all maximal job sets whose deterministic counterparts are processed contiguously without idle time in $\mathrm{S^p}|_{[k]}$ is the \emph{canonical decomposition} of $[k]$. A subset $J \subseteq [n]$ is called \emph{canonical} if $J \in \mathcal C(k)$ for some $k \in [n]$. Notice that for every canonical set~$J$ it holds that 
\begin{equation}
 J = \biggl\{j \in [n] \biggm| r_{\min}(J) \le C_j(\alpha) \le r_{\min}(J)+\sum_{j \in J} \bar p_j\biggr\}, \label{eq:canonical sets consecutive}
\end{equation}
so that the jobs in $J$ are scheduled consecutively by the $\mathrm{SOS}(\alpha)$ policy. The following \lcnamecref{thm:SOS Delta} generalizes \cite[Theorem~3.3\,(i)]{GQS+02}.

\begin{theorem} \label{thm:SOS Delta}
 Let $\alpha \in (0, 1]$ be arbitrary, and let $c \coloneqq 1+\max\bigl\{\frac 1 \alpha,\ 1+\alpha - g(\Delta) \cdot (1-\alpha)\bigr\}$. Then the $\mathrm{SOS}(\alpha)$ policy satisfies
 \[\E\biggl[\,\sum_{j=1}^n w_j \cdot \bm M_j^{\mathrm{SOS}(\alpha)}\biggr] \le c \cdot \sum_{j=1}^n w_j \cdot M_j^{\mathrm{S^p}}\le c \cdot \E\biggl[\,\sum_{j=1}^n w_j \cdot \bm M_j^{\mathrm{OPT}}\biggr].\]
\end{theorem}
\begin{proof}
 For $j \in [n-1]$ let $\sigma_j \coloneqq w_j/\bar p_j-w_{j+1}/\bar p_{j+1} \ge 0$, and let $\sigma_n \coloneqq w_n/\bar p_n$.
 We apply the following transformation to the sums on both sides of the claimed inequality.
 \[\begin{array}{>{\displaystyle}rc>{\displaystyle}cc>{\displaystyle}l}
  \E\biggl[\,\sum_{j=1}^n w_j \bm M_j^{\mathrm{SOS}(\alpha)}\biggr] &\hspace{-1.6ex}=&\hspace{-1.6ex} \sum_{j=1}^n \sum_{k=j}^n \sigma_k \bar p_j \E[\bm M_j^{\mathrm{SOS}(\alpha)}] &\hspace{-1.6ex}=&\hspace{-1.6ex} \sum_{k=1}^n \sigma_k \cdot \sum_{j=1}^k \bar p_j \E[\bm M_j^{\mathrm{SOS}(\alpha)}], \\
  \sum_{j=1}^n w_j M_j^{\mathrm{S^p}} &\hspace{-1.6ex}=&\hspace{-1.6ex} \sum_{j=1}^n \sum_{k=j}^n \sigma_k \bar p_j M_j^{\mathrm{S^p}} &\hspace{-1.6ex}=&\hspace{-1.6ex} \sum_{k=1}^n \sigma_k \cdot \sum_{j=1}^k \bar p_j M_j^{\mathrm{S^p}}.
 \end{array}\]
 As $\sigma_k \ge 0$, it suffices to prove that $\sum_{j=1}^k \bar p_j \E[\bm M_j^{\mathrm{SOS}(\alpha)}] \le c \cdot \sum_{j=1}^k \bar p_j M_j^{\mathrm{S^p}}$ for all $k \in [n]$. Since we can split both sums up according to the canonical decomposition~$\mathcal C(k)$ of the set~$[k]$, it is enough to verify the inequality
 \begin{equation} 
  \sum_{j\in J} \bar p_j \E[\bm M_j^{\mathrm{SOS}(\alpha)}] \le c \cdot \sum_{j \in J} \bar p_j M_j^{\mathrm{S^p}} \label{ineq:claim canonical sets}
 \end{equation}
 for all canonical sets~$J \subseteq [n]$.
 
 So let $J \subseteq [n]$ be an arbitrary canonical set, $J = \{j_1,\dotsc,j_\ell\}$ with $C_{j_1}(\alpha) < \cdots < C_{j_\ell}(\alpha)$, so that the jobs are scheduled in this order by the $\mathrm{SOS}(\alpha)$ policy. Let $i \in [\ell]$, and for $k \in [n]$ let $\eta_k = \eta_k(j_i) = 1/\bar p_k \cdot \int_0^{C_{j_i}(\alpha)} I_k^{\mathrm{S^p}}\dif t$. Using \cref{lem:parametric}, we can bound the expected mean busy time of $j_i$ by
 \begin{align*}
  \E[\bm C_{j_i}^{\mathrm{SOS}(\alpha)}] &\le C_{j_i}(\alpha) + \sum_{k : C_k(\alpha) \le C_{j_i}(\alpha)} \bigl(1-g(\Delta) \cdot (\eta_k - \alpha)\bigr) \cdot \bar p_k \\
  &\stackrel{\mathclap{\eqref{eq:canonical sets consecutive}}}\le C_{j_i}(\alpha) + \sum_{k : C_k(\alpha) < r_{\min}(J)} \bar p_k + \sum_{i'=1}^{i-1} \bigl(1-g(\Delta) \cdot (\eta_{j_{i'}} - \alpha)\bigr) \cdot \bar p_{j_{i'}} + \bar p_{j_i}\\
  &\le C_{j_i}(\alpha) + \frac 1 \alpha \cdot r_{\min}(J) + \sum_{i'=1}^{i-1} \bigl(1-g(\Delta) \cdot (\eta_{j_{i'}}-\alpha)\bigr) \cdot \bar p_{j_{i'}} + \bar p_{j_i},
 \end{align*}
 where we used in the last inequality that each job~$k$ with $C_k(\alpha) < r_{\min}(J)$ is processed at least $\alpha \bar p_k$ time units before $r_{\min}(J)$ and the total processing time before $r_{\min}(J)$ cannot exceed $r_{\min}(J)$. We can bound the $\alpha$-point of $j_i$ by
 \[
  C_{j_i}(\alpha) \le r_{\min}(J) + \sum_{i'=1}^{i-1} \eta_{j_{i'}} \cdot \bar p_{j_{i'}} + \alpha \cdot \sum_{i'=i}^\ell \bar p_{j_{i'}}.
 \]
 This implies that
 \begin{align*}
  &\E[\bm C_{j_i}^{\mathrm{SOS}(\alpha)}] \\
  \le{} &\Bigl(1+\frac 1 \alpha\Bigr) r_{\min}(J) + \sum_{i'=i}^{i-1} \bigl(\eta_{j_{i'}} + 1 - g(\Delta) (\eta_{j_{i'}} - \alpha)\bigr) \bar p_{j_{i'}} + \bar p_{j_i} + \alpha \sum_{i'=i}^\ell \bar p_{j_{i'}} \\
  \le{} &\Bigl(1+\frac 1 \alpha\Bigr) r_{\min}(J) + \bigl(2-g(\Delta) (1-\alpha)\bigr) \cdot \sum_{i'=1}^{i-1} \bar p_{j_{i'}} + \bar p_{j_i} + \alpha \sum_{i'=i}^\ell \bar p_{j_{i'}}.
 \end{align*}
 Let $\beta \coloneqq 2 - g(\Delta)(1-\alpha) \ge 2 - (1-\alpha) = 1+\alpha$. By subtracting $\bar p_j/2$, we obtain 
 \[
  \E[\bm M_{j_i}^{\mathrm{SOS}(\alpha)}] \le \Bigl(1 + \frac 1 \alpha \Bigr) r_{\min}(J) + \beta \sum_{i'=1}^{i-1}\bar p_{j_{i'}} + \frac{\bar p_{j_i}}{2} + \alpha \sum_{i'=i}^\ell \bar p_{j_{i'}}.
 \]
 Applying this inequality for every $j_i$, $i=1,\dotsc,\ell$, results in
 \allowdisplaybreaks
\begin{align*}
 &\sum_{j \in J} \bar p_j \E[\bm M_j^{\mathrm{SOS}(\alpha)}] = \sum_{i=1}^\ell \bar p_{j_i} \E[\bm M_{j_i}^{\mathrm{SOS}(\alpha)}] \\
 \le{} &\Bigl(1+\frac 1 \alpha\Bigr) \sum_{i=1}^\ell \bar p_{j_i} r_{\min}(J) + \beta \sum_{i=1}^\ell \sum_{i'=1}^{i-1} \bar p_{j_i} \bar p_{j_{i'}} + \frac 1 2 \sum_{i=1}^\ell \bar p_{j_i}^2 + \alpha \sum_{i=1}^\ell \sum_{i'=i}^\ell \bar p_{j_i} \bar p_{j_{i'}} \\
 ={} &\Bigl(1+\frac 1 \alpha\Bigr) \cdot \sum_{i=1}^\ell \bar p_{j_i} r_{\min}(J) + \frac{\alpha+\beta}{2} \cdot \biggl(\,\sum_{i=1}^\ell \bar p_{j_i}\biggr)^{\!2} + \frac{1+\alpha-\beta}{2} \cdot \sum_{i=1}^\ell \bar p_{j_i}^2 \\
 \le{} &\max\Bigl\{1+\frac 1 \alpha,\ \alpha + \beta\Bigr\} \cdot \sum_{j \in J} \bar p_j \cdot \biggl(r_{\min}(J) + \frac 1 2 \cdot \sum_{j \in J} \bar p_j\biggr) \\
 \le{} &\biggl(1+\max\Bigl\{\frac 1 \alpha,\ \alpha + 1 - g(\Delta) (1-\alpha) \Bigr\}\biggr) \cdot \sum_{j \in J} \bar p_{j} M_j^{\mathrm{S^p}} = c \cdot \sum_{j \in J} \bar p_j M_j^{\mathrm{S^p}},
\end{align*}
where the second inequality is valid since $1+\alpha \le \beta$ and the last inequality holds because $M^{\mathrm{S^p}}$ is a feasible solution to $(\mathrm{LP})$. So we have shown \cref{ineq:claim canonical sets} for all canonical sets, concluding the proof.
\end{proof}

Applying this \cref{thm:SOS Delta} for $\alpha \coloneqq \phi - 1$ yields \cref{prop:DSOS}. We are interested in $\alpha$ for which the \lcnamecref{thm:SOS Delta} yields the best possible performance guarantee for any given $\Delta \ge 0$. The first term in the maximum is monotonically decreasing in $\alpha$, the second term is monotonically increasing, and their graphs intersect. Consequently, the optimal $\alpha$ is attained at the intersection point, i.e., at the unique solution of the quadratic equation~$\bigl(1+g(\Delta)\bigr) \cdot \alpha^2 + \bigl(1-g(\Delta)\bigr) \cdot \alpha = 1$ in the interval $(0, 1]$. By slight abuse of notation let $\alpha \colon \R_{\ge 0} \to (0, 1]$ be the function that maps each value of $\Delta$ to the $\alpha \in (0, 1]$ yielding the best performance guarantee. Then for all $\Delta \ge 0$ it holds that
\[\alpha(\Delta) = \frac{g(\Delta) - 1 + \sqrt{g(\Delta) \cdot (g(\Delta)+2) + 5}}{2 \cdot (g(\Delta)+1)} = \begin{cases*} \frac{\sqrt \Delta - \sqrt{32-8\sqrt \Delta + \Delta}}{2 \cdot (\sqrt \Delta - 4)} &for $\Delta \le 1$; \\ \frac{\sqrt{8+12\Delta+5\Delta^2}-\Delta}{2 \cdot (2+\Delta)} &for $\Delta \ge 1$.\end{cases*}\]

\begin{corollary} \label{cor:SOS alpha Delta}
 The $\mathrm{SOS}\bigl(\alpha(\Delta)\bigr)$ policy has competitive ratio at most
 \[c \coloneqq 1 + \frac{1}{\alpha(\Delta)} = \begin{cases*} 1+ \frac{\sqrt{\Delta}+\sqrt{32-8 \cdot \sqrt \Delta + \Delta}}{4} &for $\Delta \le 1$; \\ 1 + \frac{2 \cdot (2+\Delta)}{\sqrt{8+12 \Delta +5 \Delta^2} - \Delta} &for $\Delta \ge 1$\end{cases*}\] for instances with squared coefficients of variation bounded by $\Delta$.
\end{corollary}

This converges to $\phi+1$ as $\Delta \to \infty$, and for deterministic processing times we retrieve the performance guarantee~$1+\sqrt 2$ from \cite[Theorem~3.3\,(i)]{GQS+02} for $\alpha(0) = \frac{1}{\sqrt 2}$. As in \cref{subsec:refined RSOS}, we can infer the bound for unrelated machines illustrated in \cref{fig:performance unrelated}.

\begin{corollary}
 Let $c$ be as in \cref{cor:SOS alpha Delta}. There is a deterministic online policy for unrelated machines whose competitive ratio within the class of all scheduling policies is at most $c \cdot (4 + 2\Delta) \le (\phi + 1) \cdot (4 + 2 \Delta) - \frac{5 - \sqrt 5}{10}$ and within the class of fixed-assignment policies is bounded by $4c$.
\end{corollary}

\section{Concluding remarks} \label{sec:conclusion}

We described a parameterized online policy for a single machine in \cref{subsec:delayed list scheduling}, and in \cref{subsec:refined RSOS,subsec:refined DSOS,apx:refined delta} we derived parameter choices that require prior knowledge about the processing time distributions. In this way we attenuate the pessimistic adversarial online model, where any distribution could occur at any time. The same approach was pursued by \citet{MUV06} for the derivation of their competitiveness result for $\delta$-NBUE processing times.
Let us shortly discuss the question what happens if the policy is optimized for a predicted parameter~$\bar \Delta$ (or analogously $\bar \delta$) that turns out to be wrong. \Cref{thm:RSOS Delta,thm:SOS Delta} can still be used to derive performance guarantees for $\mathrm{RSOS}(f_{\bar \Delta})$ and $\mathrm{SOS}(\alpha(\bar \Delta))$, respectively, when the maximum coefficient of variation is some $\Delta > \bar \Delta$. The worst case is when $\bar \Delta = 0$ and $\Delta \to \infty$, i.e., when the processing times are assumed to be deterministic but actually have huge coefficients of variation. In this case, the two \lcnamecrefs{thm:RSOS Delta} yield the still constant performance guarantees~$2.223$ for $\mathrm{RSOS}(f_0)$ and $2+1/\sqrt 2 \approx 2.707$ for $\mathrm{SOS}(1/\sqrt 2)$. Hence, if it is unknown whether the predicted parameter~$\bar \Delta$ is correct, one may (for the randomized policy) choose anything between the following two extrema: On the one hand, choosing $\bm \alpha_j$ according to the uniform distribution, as in \cref{prop:RSOS}, always leads to $2$-competitiveness. On the other hand, using the probability density function~$f_0$, as in \cite{GQS+02}, leads to $1.686$-competitiveness for deterministic processing times but only to $2.223$-competitiveness in case the processing times have unbounded coefficients of variation. For the choice of $\alpha$ in the deterministic policy an analogous trade-off occurs . The exact bounds resulting from all possible combinations of $\bar \Delta$ and $\Delta$ are illustrated in \cref{fig:untrusted}.
\begin{figure}
 \begin{subfigure}{\linewidth}
  \centering
  \includegraphics[width=.9\textwidth]{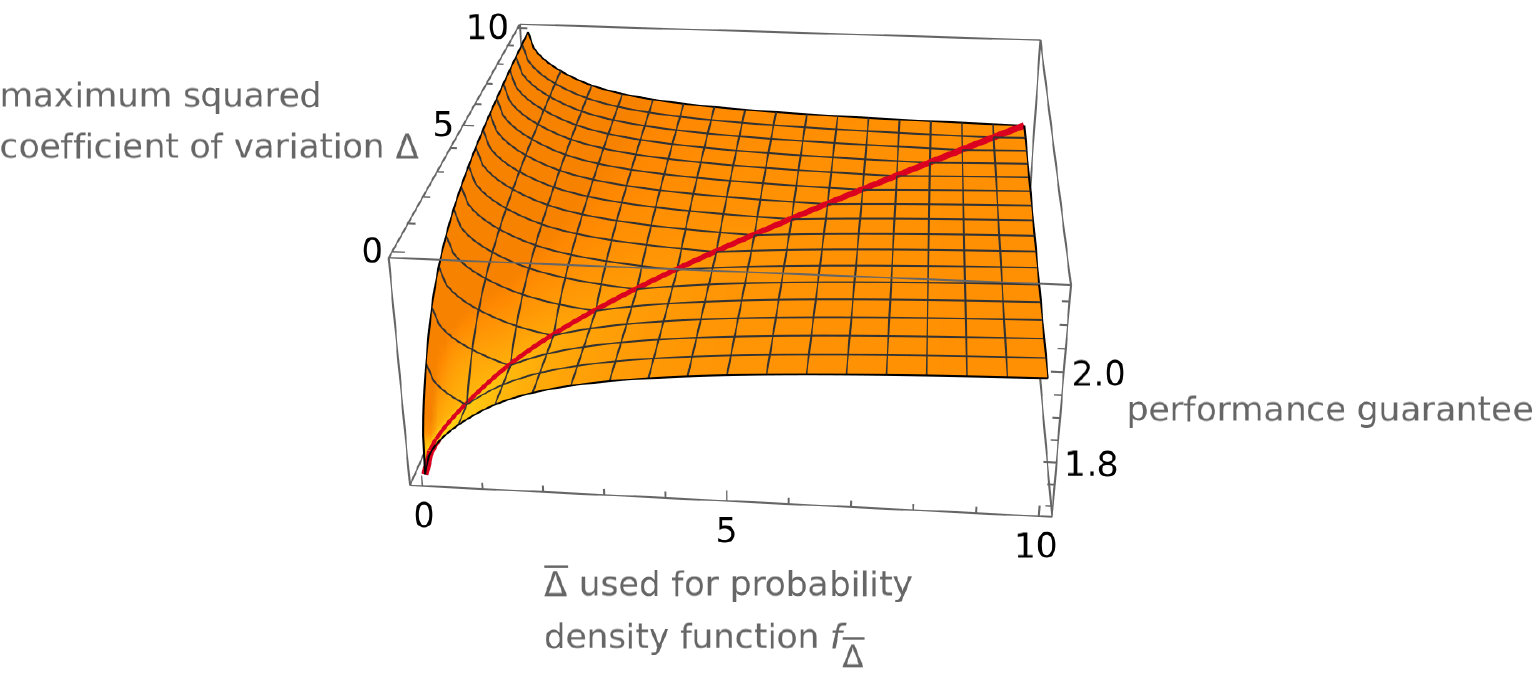}
  \caption{$\mathrm{RSOS}(f_{\bar \Delta})$}
 \end{subfigure}
 \begin{subfigure}{\linewidth}
  \centering
  \includegraphics[width=.9\textwidth]{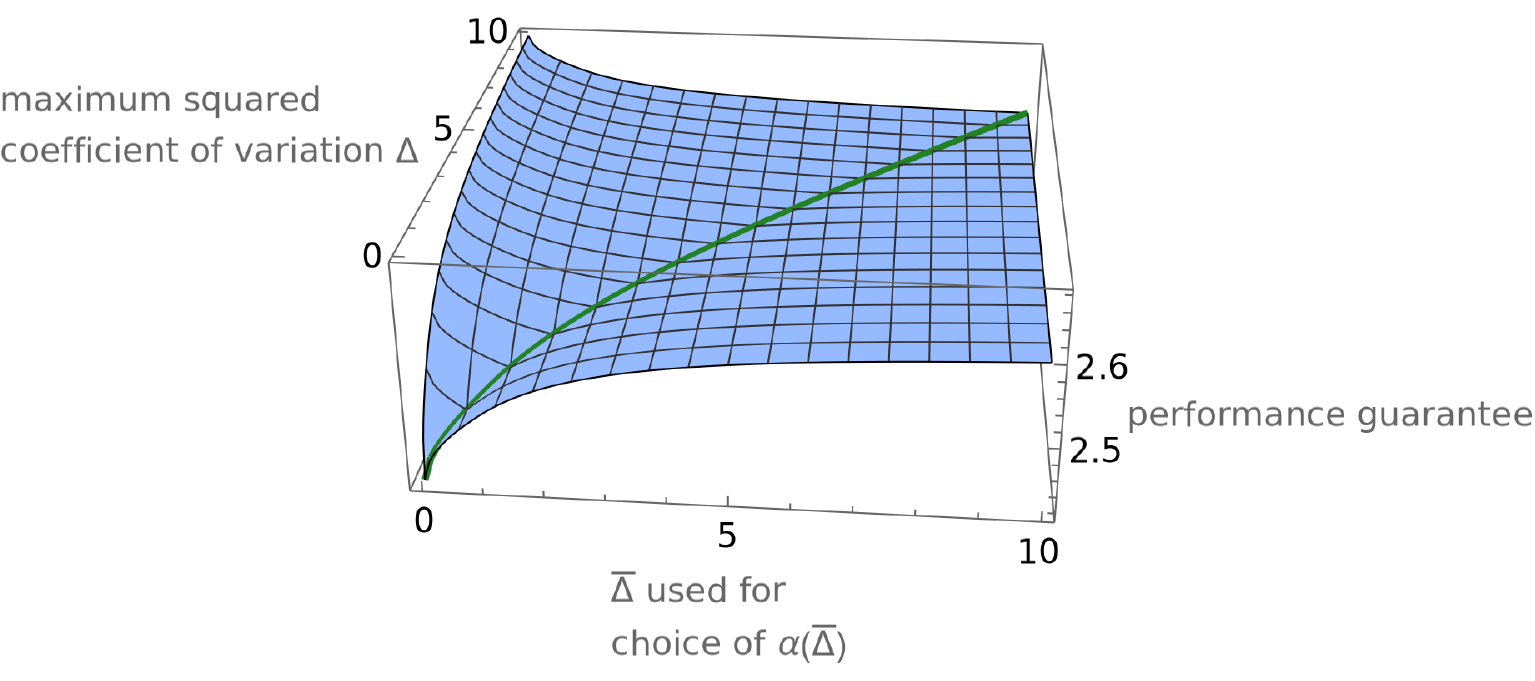}
  \caption{$\mathrm{SOS}(\alpha(\bar \Delta))$}
 \end{subfigure}
 \caption{Performance guarantees for online policies adapted to $\bar \Delta$ for instances with maximum squared coefficient of variation~$\Delta$. For every $\Delta$ the best guarantee is achieved for $\bar \Delta = \Delta$. This is indicated by the red and green lines, which coincide with the corresponding curves in \cref{fig:performance single}.}
 \label{fig:untrusted} 
\end{figure}
This relates our results to the model of untrusted predictions, which has recently attracted a lot of attention~\cite{PSK18}. Unlike most results in this area, we do not need an additional parameter to factor in the confidence into the prediction, but we can incorporate lower confidence by simply using a greater value~$\bar \Delta$ than predicted.

In \cref{seq:refined policies}, we generalized the algorithms of \citet{GQS+02} for deterministic single-machine scheduling to stochastic scheduling. Another generalization of these algorithms was provided by \citet{CW09}, who considered identical parallel machines. The parallelism of our and their work is particularly evident for the randomized online algorithm, where both extensions lead to a competitive ratio of $2$ for arbitrary instances, which can be enhanced for each specific value of the added problem parameter---in our case the upper bound~$\Delta$ on the squared coefficients of variation, and in their case the number~$m$ of machines. In both generalizations there is one parameter value that corresponds to the basic case considered by \citeauthor{GQS+02} and for which the enhanced algorithm and analysis coincide exactly with the original algorithm and analysis. Although both works lead to a competitive ratio of $2$, for their combination only the non-constant performance guarantee~$2+\Delta$ of \citet{Sch08} is known. Our and \citeauthor{CW09}'s results together yield a refinement of the RSOS policy for all pairs $m$, $\Delta$ with $m = 1$ or $\Delta = 0$.

Finally, \cref{sec:RSOS DSOS} implies that there is a constant-competitive online delayed list-scheduling policy for a single machine with release dates. The delays are only necessary in the online context, so there is also a usual job-based list scheduling policy with a constant performance guarantee. For identical parallel machines, however, for each of the following policy classes it is not known whether a constant performance guarantee can be achieved: online policies, (delayed) list-scheduling policies, and efficiently executable policies. The only known classes of policies for which this can be excluded are fixed-assignment policies and index policies~\cite{EFMM19}, a subclass of machine-based list-scheduling policies. All the open cases are interesting research questions. For example, is there an online policy that can take arbitrary computation time and has a constant competitive ratio?

\bibliographystyle{elsarticle-num-names}
\bibliography{bibliography_arxiv}

\appendix\allowdisplaybreaks

\section{Omitted lemmas and proofs} \label{apx:technical}
 
\begin{lemma} \label{lem:expected positive part Delta}
 Let $\bm X$ be a non-negative random variable with $\CV[\bm X]^2 \le \Delta$, and let $\beta \in [0, 1)$. Then
 \[\E\bigl[(\bm X - \beta \cdot \E[\bm X])^+\bigr] \le \bigl(1-g(\Delta) \cdot \beta\bigr) \cdot \E[\bm X].\]
\end{lemma}
\begin{proof} 
 We do a case distinction:
 \begin{itemize}
  \item Consider first the case that $\Delta \le 1$. Then we have realizationwise
  \[\bm X - \beta \cdot \E[\bm X] \le \bm X - \beta \bm X + \beta \cdot \bigl(\bm X - \E[\bm X]\bigr)^+ = (1-\beta) \cdot \bm X + \beta \cdot \bigl(\bm X-\E[\bm X]\bigr)^+.\]
  Since the right-hand side is non-negative, it also holds that \[\bigl(\bm X-\beta \cdot \E[\bm X]\bigr)^+ \le (1-\beta) \cdot \bm X + \beta \cdot \bigl(\bm X-\E[\bm X]\bigr)^+.\] Therefore, using that $\E\bigl[(\bm X - \E[\bm X])^+\bigr] = \frac 1 2 \E\bigl[\lvert\bm X - \E[\bm X]\rvert\bigr]$, we obtain
  \begin{multline*}\E\bigl[(\bm X - \beta \E[\bm X])^+\bigr] \le (1-\beta) \cdot \E[\bm X] + \frac \beta 2 \cdot \E\bigl[\lvert\bm X - \E[\bm X]\rvert\bigr]\\ \le (1-\beta) \cdot \E[\bm X] + \frac \beta 2 \cdot \sqrt{\E\bigl[(\bm X - \E[\bm X])^2\bigr]} \le \biggl(1-\beta \cdot \frac{2-\sqrt{\Delta}}{2}\biggr) \cdot \E[\bm X].
  \end{multline*}
  \item Consider now the case~$\Delta \ge 1$. Then we have realizationwise
  \begin{align*}
   &\bm X-\beta \cdot \E[\bm X] \le \bm X - \beta \cdot \biggl(\E[\bm X] - \frac{\bigl(\bm X-\E[\bm X] \cdot (\Delta+1)\bigr)^2}{\E[\bm X] \cdot (\Delta + 1)^2}\biggr)\\ ={} 
   &\bm X - \beta \cdot \biggl(\frac{2\bm X}{\Delta+1} -\frac{\bm X^2}{\E[\bm X] \cdot (\Delta+1)^2}\biggr) \\
   ={} &\Bigl(1-\frac{2\beta}{\Delta+1}\Bigr) \cdot \bm X + \frac{\beta}{\E[\bm X] \cdot (\Delta+1)^2} \cdot \bm X^2.
  \end{align*}
  Since the right-hand side is non-negative, it also holds that \[\bigl(\bm X-\beta \cdot \E[\bm X]\bigr)^+ \le \Bigl(1-\frac{2\beta}{\Delta+1}\Bigr) \cdot \bm X + \frac{\beta}{\E[\bm X] \cdot (\Delta+1)^2} \cdot \bm X^2.\] Consequently, using that $\E[\bm X^2] \le (\Delta+1) \cdot \E[\bm X]^2$, we obtain
  \begin{align*}
   \E\bigl[(\bm X-\beta \cdot \E[\bm X])^+\bigr] &\le \Bigl(1-\frac{2\beta}{\Delta+1}\Bigr) \cdot \E[\bm X] + \frac{\beta}{\Delta+1} \cdot \E[\bm X] \\
   &= \Bigl(1-\frac{\beta}{\Delta+1}\Bigr) \cdot \E[\bm X]. \qedhere
  \end{align*}
 \end{itemize}
\end{proof}

\Gooddensity*

\begin{proof}
 Recall that $0 < D = g(\Delta) \le 1$. For $\gamma = 0$ the left-hand side is 
 \[1-D+{}\underbrace{D \cdot (1+D)}_{>0}{} \cdot \bigl(\underbrace{\ln(1+D)-D}_{<0}\bigr) < 1-D,\]
 and for $\gamma=1$ the left-hand side is equal to 
 \[1-D + {}\underbrace{\mathrm e^D}_{>0}{} \cdot \bigl(\underbrace{\mathrm e^{-D} - (1-D)}_{>0}\bigr) \cdot (\underbrace{\mathrm e^D + D - 1}_{>0}) > 1-D.\]
 As the left-hand side is continuous in $\gamma$ for $\gamma \le 1$, the intermediate value theorem implies the existence of a solution $\gamma$ with $0 < \gamma < 1$.
 
 It holds that $\mathrm e^{D \theta} = \mathrm e^{D \gamma} \cdot \bigl(1+D \cdot (1-\gamma)\bigr) > 1$, so that $c > 1$. Therefore, $f(\alpha) \ge 0$ for all $\alpha \in (0, 1]$. Furthermore,
 \[\int_0^1 f(\alpha)\dif \alpha = \int_0^\theta (c-1) \cdot \mathrm e^{D \alpha}\dif \alpha = \frac{c-1}{D} \cdot (\mathrm e^{D \theta} - 1) = 1,\]
 that is, $f$ is indeed a probability density function on $(0, 1]$. We show that it fulfills the two conditions from \cref{thm:RSOS Delta}. We start with the first condition. Let $x \in [0, 1]$. We start with the case that $x \le \theta$.
 \begin{multline*}
  \int_0^x \bigl(1-D \cdot (x-\alpha)\bigr) \cdot f(\alpha)\dif \alpha = \int_0^x \bigl(1-D \cdot (x-\alpha)\bigr) \cdot (c-1) \cdot \mathrm e^{D \alpha}\dif \alpha \\
  = \frac{c-1}{D} \cdot \bigl(\mathrm e^{D x} - (1-D x) - \mathrm e^{D x} + 1\bigr) = (c-1) \cdot x.
 \end{multline*}
 Now consider $x \ge \theta$. Then
 \begin{multline*}
  \int_0^x \bigl(1-D \cdot (x-\alpha)\bigr) \cdot f(\alpha)\dif \alpha = \int_0^\theta \bigl(1-D \cdot (x-\alpha)\bigr) \cdot (c-1) \cdot \mathrm e^{D \alpha}\dif \alpha \\
  = \frac{c-1}{D} \cdot \bigl((1 - D \cdot (x-\theta)) \cdot \mathrm e^{D \theta} - (1-D x) - \mathrm e^{D \theta} + 1\bigr) \le (c-1) \cdot x.
 \end{multline*}
 We now prove the second condition. It holds that
 \begin{align*}
  \int_0^1 \alpha \cdot f(\alpha)\dif \alpha &= \int_0^\theta \alpha \cdot (c-1) \cdot \mathrm e^{D \alpha}\dif \alpha \\
  &= \frac{c-1}{D} \cdot \Bigl(\theta \cdot \mathrm e^{D \theta} - \frac{1}{D} \cdot (\mathrm e^{D \theta} - 1)\Bigr) 
  = \frac{\theta \cdot \mathrm e^{D \theta}}{\mathrm e^{D \theta}-1}- \frac{1}{D}
 \end{align*}
 Consequently,
 \[2-D \cdot \biggl(1-\int_0^1 \alpha \cdot f(\alpha)\dif \alpha\biggr) = 1 - D \cdot \biggl(1-\frac{\theta \cdot \mathrm e^{D \theta}}{\mathrm e^{D \theta} - 1}\biggr) > 0.\]
 For $x \le 1-\theta$ the second condition is trivial. Consider now $x \ge 1 - \theta$. Then
 \[\int_{1-x}^1 f(\alpha)\dif \alpha = \int_{1-x}^\theta (c-1) \cdot \mathrm e^{D \alpha}\dif \alpha = \frac{c-1}{D} \cdot \bigl(\mathrm e^{D \theta} - \mathrm e^{D (1-x)}\bigr) = \frac{\mathrm e^{D \theta} - \mathrm e^{D (1-x)}}{\mathrm e^{D \theta} - 1}.\]
 We combine this and obtain
 \begin{align*}
  &\Biggl(2-D \cdot \biggl(1-\int_0^1 \alpha \cdot f(\alpha)\dif \alpha\biggr)\Biggr) \cdot \int_{1-x}^1 f(\alpha)\dif \alpha\\
  ={} &\Biggl(1-D \cdot \biggl(1 - \frac{\theta \cdot \mathrm e^{D \theta}}{\mathrm e^{D \theta} - 1}\biggr)\Biggr) \cdot \frac{\mathrm e^{D \theta} - \mathrm e^{D \cdot (1-x)}}{\mathrm e^{D \theta} - 1} \\
  ={} &\Biggl(1-D \cdot \biggl(1 - \frac{\theta \cdot \mathrm e^{D \theta}}{\mathrm e^{D \theta} - 1}\biggr)\Biggr) \cdot \frac{\mathrm e^{D \theta} - \mathrm e^{D \gamma} \cdot \mathrm e^{D \cdot (1-\gamma-x)}}{\mathrm e^{D \theta} - 1} \\
  \le{} &\Biggl(1-D \cdot \biggl(1 - \frac{\theta \cdot \mathrm e^{D \theta}}{\mathrm e^{D \theta} - 1}\biggr)\Biggr) \cdot \frac{\mathrm e^{D \theta} - \mathrm e^{D \gamma} \cdot \bigl(1+D \cdot (1-\gamma-x)\bigr)}{\mathrm e^{D \theta} - 1} \\
  ={} &\Biggl(1-D \cdot \biggl(1 - \frac{\theta \cdot \mathrm e^{D \theta}}{\mathrm e^{D \theta} - 1}\biggr)\Biggr) \cdot \frac{\mathrm e^{D \gamma} \cdot D x}{\mathrm e^{D \theta} - 1} \\
  ={} &\left(1 + \frac{\bigl(1-D \cdot \bigl(1 - \frac{\theta \cdot \mathrm e^{D \theta}}{\mathrm e^{D \theta} - 1}\bigr)\bigr) \cdot \mathrm e^{D \gamma} \cdot D - (\mathrm e^{D \theta} - 1)}{\mathrm e^{D \theta} - 1}\right) \cdot x,
 \end{align*}
 where we used the definition of $\theta$ in the second last equation. We claim that this equals $c \cdot x = \bigl(1+\frac{D}{\mathrm e^{D \theta}-1}\bigr) \cdot x$. In order show this, it remains to be shown that the numerator equals $D$. The following computation thus concludes the proof.
 \footnotesize
 \begin{align*}
  &\Biggl(1 - D \cdot \biggl(1 - \frac{\theta \cdot \mathrm e^{D \theta}}{\mathrm e^{D \theta} - 1}\biggr)\Biggr) \cdot \mathrm e^{D \gamma} D - (\mathrm e^{D\theta} - 1) \\
  ={} &\biggl(1 - D + D \cdot \frac{\theta \cdot \mathrm e^{D \theta}}{\mathrm e^{D \theta} - 1}\biggr) \cdot \mathrm e^{D \gamma} D - \mathrm e^{D\gamma} \cdot \bigl(1+D (1-\gamma)\bigr) + 1 \\
  ={} &\mathrm e^{D \gamma} \cdot \Biggl(D \cdot \biggl(\gamma - D + D \cdot \frac{\theta \cdot \mathrm e^{D \theta}}{\mathrm e^{D \theta} - 1}\biggr)-1\Biggr) + 1\\
  ={} &\mathrm e^{D \gamma} \cdot \biggl(D \cdot \frac{\mathrm{e}^{D\theta} \cdot \bigl(\gamma - D \cdot (1-\theta)\bigr) + D - \gamma}{\mathrm e^{D \theta} - 1} - 1 \biggr) + 1 \\ 
  ={} &\mathrm e^{D \gamma} \cdot \frac{\mathrm e^{D \theta} \cdot \bigl(D \cdot (\gamma - D \cdot (1-\theta))-1\bigr)+D \cdot (D-\gamma) + 1}{\mathrm e^{D \theta} - 1} + 1 \\
  ={} &\mathrm e^{D \gamma} \frac{\mathrm e^{D \gamma} \bigl(1+D (1-\gamma)\bigr) \bigl(D (\gamma - D (1 - \gamma) + \ln(1 + D (1-\gamma))) - 1\bigr) + D^2 - D \gamma}{\mathrm e^{D\theta} - 1} + 1 \\
  ={} &\frac{1-D+\mathrm e^{D \gamma} \cdot \bigl(D \cdot (\gamma + D \cdot (1-\gamma)) - 1\bigr)}{\mathrm e^{D\theta} - 1} + 1 \\
  ={} &\frac{\mathrm e^{D\theta} + \mathrm e^{D\gamma} \cdot \bigl(D \cdot (\gamma + D \cdot (1-\gamma)) - 1\bigr) - D}{\mathrm e^{D\theta} - 1} \\
  ={} &\frac{e^{D \gamma} \cdot D \cdot \bigl(1 + D \cdot (1-\gamma)\bigr) - D}{\mathrm e^{D\theta} - 1} 
  = \frac{D \cdot (\mathrm e^{D\theta} - 1)}{\mathrm e^{D\theta} - 1} = D. \tag*{\normalsize\qedhere}
 \end{align*}
\end{proof}

\section{Refined policies for \texorpdfstring{$\delta$}{delta}-NBUE processing times} \label{apx:refined delta}

In this \lcnamecref{apx:refined delta}, we always assume that processing times are $\delta$-NBUE for some $\delta \ge 1$. Since $\Delta \le 2 \delta + 1$, we can apply the results from \cref{seq:refined policies} in order to improve upon the RSOS and DSOS policy in this case. However, a more direct adaption of the policies and analyses to this case is possible, leading to better performance guarantees, see \cref{fig:performance single delta}.
\begin{figure}
 \centering
 \begin{tikzpicture}[font=\footnotesize]
  \begin{axis}[
    xmin=1, xmax=2.05, ymin=1, ymax=2.8, xlabel={$\delta$}, ylabel={performance guarantee},
    ytick={1, 1.5, 1.7833, 2, 2.452, {(3+sqrt(5))/2}},
    yticklabels={$1$, $1.5$, $1.7833$, $2$, $2.452$, $\phi+1$}
   ]
   \addplot[myblue, very thick] coordinates {(1, {(3+sqrt(5))/2}) (2, {(3+sqrt(5))/2})} node[yshift=10pt, right] {deterministic (\Cref{prop:DSOS})};
   \addplot[mygreen, very thick, domain=1:2] plot (x, {1+(2*(2+(2*x-1)))/(sqrt(8+12*(2*x-1)+5*(2*x-1)*(2*x-1))-(2*x-1))}) node[anchor=north west, align=left, yshift=15pt] {deterministic, optimized for\\[-3pt]$\Delta = 2 \delta - 1$ (\Cref{cor:SOS alpha Delta})};
   \addplot[very thick, lime!75!mygreen] table[x expr=1+\coordindex/1000, y index=0] {det_single_machine_delta.dat} node[yshift=-1pt, align=left, anchor=north west] {deterministic, optimized directly\\[-3pt]for $\delta$};
 
   \addplot[myorange, very thick] coordinates {(1, 2) (2, 2)} node [yshift=10pt, right] {randomized (\Cref{prop:RSOS})};
   \addplot[very thick, myred] table[x expr=1+\coordindex/1000, y index=0] {rand_single_machine_2delta-1.dat} node[yshift=15pt, align=left, anchor=north west] {randomized, optimized for\\[-3pt]$\Delta = 2\delta - 1$ (\Cref{cor:RSOS density Delta})};
   \addplot[very thick, magenta] table[x expr=1+\coordindex/1000, y index=0] {rand_single_machine_delta.dat} node[yshift=-1pt, align=left, anchor=north west] {randomized, optimized directly\\[-3pt]for $\delta$};
  \end{axis}
 \end{tikzpicture}
 \caption{Performance guarantees of different online policies for a single machine for $\delta$-NBUE processing times as a function of $\delta$. }
 \label{fig:performance single delta}
\end{figure}
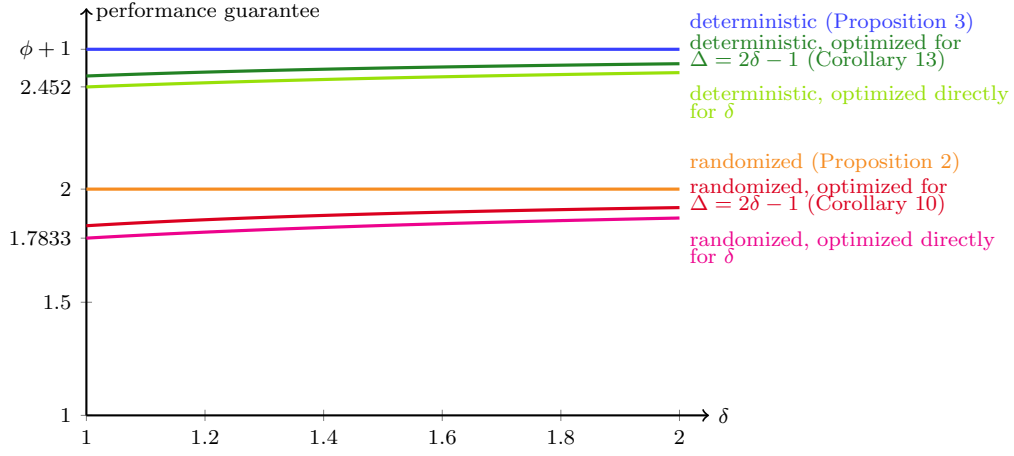
The main insight for this is the following analogue of \cref{lem:expected positive part Delta}.

\begin{lemma} 
 Let $\bm X$ be a $\delta$-NBUE non-negative random variable for some $\delta \ge 1$, and let $\beta \in [0, 1)$. Then
 \[\E\bigl[(\bm X - \beta \cdot \E[\bm X])^+\bigr] \le \frac{\delta}{\delta + \beta} \cdot \E[\bm X].\]
\end{lemma}
\begin{proof}
 \[\E\bigl[(\bm X - \beta \cdot \E[\bm X])^+\bigr] = \Pr\bigl[\bm X > \beta \cdot \E[\bm X]\bigr] \cdot \condexp[\big]{\bm X - \beta \cdot \E[\bm X]}{\bm X > \beta \cdot \E[\bm X]}.\]
 We show the claim by case distinction:
 \begin{itemize}
  \item If $\Pr\bigl[\bm X > \beta \cdot \E[\bm X]\bigr] \le \frac{1}{\delta + \beta}$, then the claim holds because the processing times are $\delta$-NBUE.
  \item If $\Pr\bigl[\bm X > \beta \cdot \E[\bm X]\bigr] > \frac{1}{\delta + \beta}$, then
  \begin{align*}
   &\Pr\bigl[\bm X > \beta \cdot \E[\bm X]\bigr] \cdot \E\bigl[\bm X - \beta \cdot \E[\bm X] \mid \bm X > \beta \cdot \E[\bm X]\bigr] \\
   ={} &\Pr\bigl[\bm X > \beta \cdot \E[\bm X]\bigr] \cdot \bigl(\condexp{\bm X}{\bm X > \beta \cdot \E[\bm X]} - \beta \cdot \E[\bm X]\bigr) \\
   \le{} &\E[\bm X] - \Pr\bigl[\bm X > \beta \cdot \E[\bm X]\bigr] \cdot \beta \cdot \E[\bm X] \\
   <{} &\Bigl(1-\frac{1}{\delta+\beta} \cdot \beta\Bigr) \cdot \E[\bm X] = \frac{\delta}{\delta+\beta} \cdot \E[\bm X],
  \end{align*}
  where the ``$\le$'' follows from the law of total probability and the non-negativity of $\bm X$. \qedhere
 \end{itemize}
\end{proof}

Analogously to \cref{lem:parametric}, together with \cref{lem:bound completion time SOS}, this implies that for each job~$j \in [n]$ and all $A \in (0, 1]^{\N}$ it holds that
\begin{equation}
 \E[\bm C_j^{\mathrm{SOS}(A)}] \le C_j(\alpha_j) + \sum_{k:\alpha_k \le \eta_k} \frac{\delta}{\delta + \eta_k - \alpha_k} \cdot \bar p_k. \label{ineq:upper bound expected completion time SOS(A) delta}
\end{equation}

\begin{theorem} 
 Assume that all processing times are $\delta$-NBUE for some $\delta \ge 1$. Let $f$ be a probability density function on $(0, 1]$, and let $c \in \R$ be such that for every $x \in (0, 1]$ the following two conditions are fulfilled:
 \begin{enumerate}
  \item $\int_0^x \frac{\delta}{\delta+x-\alpha} \cdot f(\alpha)\dif \alpha \le (c-1) \cdot x$,
  \item $\left(1+\int_0^1 \frac{\delta}{\delta+1-\alpha} \cdot f(\alpha)\dif \alpha\right) \cdot \int_{1-x}^1 f(\alpha)\dif \alpha \le c \cdot x$.
 \end{enumerate}
 Then the $\mathrm{RSOS}(f)$ policy satisfies inequality~\eqref{ineq:RSOS(f)}.
\end{theorem}
\begin{proof}
 The proof proceeds analogously to the proof of \cref{thm:RSOS Delta}. Let $j \in [n]$, and let $N_1$ and $N_2$ as well as $\eta_k$, $k \in [n]$, and $\mu_k$, $k \in N_2$, be as defined in \cref{subsec:refined RSOS} with respect to this job~$j$.
 Substituting \cref{eq:alpha point} in \cref{ineq:upper bound expected completion time SOS(A) delta} results in
 \begin{equation}
  \E[\bm C_j^{\mathrm{SOS}(\mathrm A)}] \le S_j^{\mathrm{S^p}} + \sum_{\substack{k \in N_1 \\ \alpha_k \le \eta_k}} \frac{\delta}{\delta+\eta_k-\alpha_k} \bar p_k + \sum_{\substack{k \in N_2 \\ \alpha_j > \mu_k}} \Bigl(1+\frac{\delta}{\delta + 1 - \alpha_k}\Bigr) \bar p_k + (1+\alpha_j) \bar p_j. \label{ineq:upper bound expected completion time SOS(A) delta proof RSOS}
 \end{equation}
 By the law of total expectation and the independence of $\bm p_j$, $j \in [n]$, and $\bm \alpha_j$, $j \in [n]$,
 \begin{align*}
  &\E[\bm C_j^{\mathrm{RSOS}(f)}] = \int_0^1 \cdots \int_0^1 \E[\bm C_j^{\mathrm{SOS}(\alpha_1,\dotsc,\alpha_n)}] f(\alpha_1) \dotsm f(\alpha_n)\dif \alpha_1 \cdots \dif \alpha_n \\
  \stackrel{\mathclap{\eqref{ineq:upper bound expected completion time SOS(A) delta proof RSOS}}}\le{} &S_j^{\mathrm{S^p}} + \sum_{k \in N_1} \int_0^{\eta_k} \frac{\delta}{\delta+\eta_k-\alpha_k} \cdot f(\alpha_k) \dif \alpha_k \bar p_k\\ &\phantom{S_j^{\mathrm{S^p}}}{} + \sum_{k \in N_2} \biggl(1+\int_0^1 \frac{\delta}{\delta+1-\alpha_k} \cdot f(\alpha_k) \dif \alpha_k\biggr) \cdot \int_{\mu_k}^1 f(\alpha_j)\dif\alpha_j \bar p_k\\ &\phantom{S_j^{\mathrm{S^p}}}{} + \biggl(1 + \int_0^1 \alpha_j \cdot f(\alpha_j)\dif \alpha_j\biggr) \cdot \bar p_j \\
  \le{} &S_j^{\mathrm{S^p}} + \sum_{k \in N_1} (c-1) \eta_k \bar p_k + \sum_{k \in N_2} c (1-\mu_k) \bar p_k \\
  &\phantom{S_j^{\mathrm{S^p}}}{} + \biggl(1+\int_0^1 \frac{\delta}{\delta+1-\alpha_j} \cdot f(\alpha_j)\dif \alpha_j\biggr) \cdot \bar p_j \\
  \stackrel{\mathclap{\eqref{ineq:total processing time before S_j}}}\le{} &c \cdot S_j^{\mathrm{S^p}} + c \cdot \sum_{k \in N_2} (1-\mu_k) \bar p_k + c \cdot \bar p_j \\
  ={} &c \cdot \biggl(\int_0^1 S_j^{\mathrm{S^p}} + \sum_{\substack{k \in N_2\\ \alpha > \mu_k}} \bar p_k + \alpha \bar p_j\dif \alpha + \frac{\bar p_j}{2}\biggr) \\
  \stackrel{\mathclap{\eqref{eq:alpha point}}}={} &c \cdot \biggl(\int_0^1 C_j(\alpha)\dif \alpha + \frac{\bar p_j}{2}\biggr) \stackrel{\eqref{eq:mean busy time alpha point}}= c \cdot \Bigl(M_j^{\mathrm{S^p}} + \frac{\bar p_j}{2}\Bigr),
 \end{align*}
 where we used in the second inequality that $\alpha_j \le \frac{\delta}{\delta+1-\alpha_j}$ for all $\alpha_j \in [0, 1]$ and $\delta \ge 1$. Consequently, by linearity of expectation and \cref{lem:lower bound},
 \[\E\biggl[\,\sum_{j=1}^n w_j \bm C_j^{\mathrm{RSOS}(f)}\biggr] \le c \cdot \sum_{j=1}^n w_j \cdot \Bigl(M_j^{\mathrm{S^p}} + \frac{\bar p_j}{2}\Bigr) \le c \cdot \E\biggl[\,\sum_{j=1}^n w_j \bm C_j^{\mathrm{OPT}}\biggr]. \qedhere\]
\end{proof}
To get the results in \cref{fig:performance single delta}, we numerically computed a good probability density functions and resulting performance guarantees for $\delta$-NBUE processing times for $\delta \in [1, 2]$. This is done by approximating the density using step functions on a discretization of $(0, 1]$ by solving a quadratically constrained program using Gurobi, for details cf.~\cite{diss}. The resulting probability density function for $\delta = 1$ is illustrated in \cref{fig:density delta}.
\begin{figure}
 \centering
 \begin{tikzpicture}[font=\footnotesize]
  \begin{axis}[
    xmin=0, xmax=1.05, ymin=0, ymax=1.39, xlabel={$\alpha$}, ylabel={$f(\alpha)$}
   ]
   \addplot[very thick, jump mark left] table[x expr=\coordindex/1000, y index=0] {density_delta.dat};
  \end{axis} 
 \end{tikzpicture}
 \caption{Probability density function~$f$ on $(0, 1]$ used for NBUE processing times.}
 \label{fig:density delta}
\end{figure}
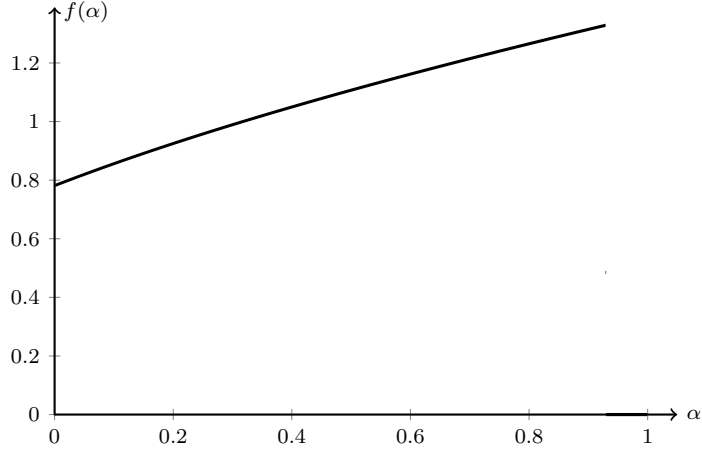

In the same way, the proof of \cref{thm:SOS Delta} can be modified by replacing the application of \cref{lem:parametric} with \cref{ineq:upper bound expected completion time SOS(A) delta}. This yields for $\alpha \in (0, 1]$ the performance guarantee \[c \coloneqq \max\Bigl\{1+\frac 1 \alpha,\ \frac{(2+\alpha) \cdot \delta+1 - \alpha^2}{\delta+1-\alpha}\Bigr\}.\]
For any given $\delta \ge 0$ the first term is monotonically increasing in $\alpha$, and the second term is decreasing. Therefore, the maximum is minimized at the unique solution fo the cubic equation $\alpha^3 = (1+\delta) \cdot (\alpha^2 + \alpha - 1)$ in the interval~$(0, 1]$. The resulting performance guarantees are shown in \cref{fig:performance single delta}.

Note that the parameter $\delta$ cannot be bounded from above in terms of $\Delta$.
This is shown by the simple example of a random variable $\bm X$ that takes the two values $1$ and $N \in \N$ both with probability~$\frac 1 2$. Therefore, the policies and competitiveness bounds for $\delta$-NBUE processing times developed here cannot be used to derive bounds for processing times with bounded coefficients of variation.

\end{document}